\documentclass[12pt,letterpaper]{article}

\usepackage{fixltx2e}
\usepackage{textcomp}
\usepackage{fullpage}
\usepackage{amsfonts}
\usepackage{verbatim}
\usepackage[english]{babel}
\usepackage{pifont}
\usepackage{color}
\usepackage{setspace}
\usepackage{lscape}
\usepackage{indentfirst}
\usepackage[normalem]{ulem}
\usepackage{booktabs}
\usepackage{natbib}
\usepackage{float}
\usepackage{latexsym}
\usepackage{url}
\usepackage{hyperref}
\usepackage{epsfig}
\usepackage{graphicx}
\usepackage{amssymb}
\usepackage{amsmath}
\usepackage{bm}
\usepackage{array}
\usepackage[version=3]{mhchem}
\usepackage{ifthen}
\usepackage{caption}
\usepackage{hyperref}
\usepackage{amsthm}
\usepackage{amstext}
\usepackage{lineno}

\def\L{\mathcal{L}}
\def\T{\mathcal{T}}

\newcommand{\PP}{\mbox{$\mathbb P$}}
\newcommand{\EE}{\mbox{$\mathbb E$}}
\newtheorem{theorem}{Theorem}
\newtheorem{proposition}[theorem]{Proposition}

\newtheorem{lemma}[theorem]{Lemma}

\raggedright
\renewcommand{\section}[1]{%
\bigskip
\begin{center}
\begin{Large}
\normalfont\scshape #1
\medskip
\end{Large}
\end{center}}

\renewcommand{\subsection}[1]{%
\bigskip
\begin{center}
\begin{large}
\normalfont\itshape #1
\end{large}
\end{center}}

\renewcommand{\subsubsection}[1]{%
\vspace{2ex}
\noindent
\textit{#1.}---}

\renewcommand{\tableofcontents}{}

\bibpunct{(}{)}{;}{a}{}{,}  

\begin{document}
\begin{flushright}
Version dated: \today
\end{flushright}
\bigskip
\noindent Predicting ancestral states in a tree

\bigskip
\medskip
\begin{center}

\noindent{\Large \bf Predicting the ancestral character changes in a tree is typically easier than predicting the root state}
\bigskip



\noindent {\normalsize \sc Olivier Gascuel$^1$ and Mike Steel$^2$}\\
\noindent {\small \it 
$^1$Institut de Biologie Computationnelle, LIRMM, CNRS \& Universit{\'e} de Montpellier, France;\\
$^2$Allan Wilson Centre for Molecular Ecology and Evolution, University of Canterbury, Christchurch, New Zealand}\\
\end{center}
\medskip
\noindent{\bf Corresponding author:} Mike Steel, University of Canterbury, Christchurch, New Zealand; E-mail: mike.steel@canterbury.ac.nz\\


\vspace{1in}

\subsubsection{Abstract}Predicting the ancestral sequences of a group of homologous sequences related by a phylogenetic tree has been the subject of many studies, and numerous methods have been proposed to this purpose. Theoretical results are available that show that when the mutation rate become too large, reconstructing the ancestral state at the tree root is no longer feasible. Here, we also study the reconstruction of the ancestral changes that occurred along the tree edges. We show that, depending on the tree and branch length distribution, reconstructing these changes (i.e.  reconstructing the ancestral state of all internal nodes in the tree) may be easier or harder than reconstructing the ancestral root state. However, results from information theory indicate that for the standard Yule tree,  the task of reconstructing internal node states remains feasible, even for very high substitution rates. Moreover, computer simulations demonstrate that for more complex trees and scenarios, this result still holds. For a large variety of counting, parsimony-based and likelihood-based methods, the predictive accuracy of a randomly selected internal node in the tree is indeed much higher than the accuracy of the same method when applied to the tree root. Moreover, parsimony- and likelihood-based methods appear to be remarkably robust to sampling bias and model mis-specification.\\
\noindent (Keywords: Ancestral state prediction, character evolution, phylogenetic tree, Markov model)\\


\vspace{1.5in}
\addtocounter{section}{1}

\newpage

\section{Introduction}

A fundamental challenge in evolutionary biology is to understand how the traits we observe today in different species evolved from some common ancestral state.
A phylogenetic tree linking the species in question provides the usual way to study this question \citep{lib}. With a tree, one  can attempt to reconstruct the evolution of the traits that we observe at the leaves of the tree by estimating the ancestral
state at the root of the tree and at the other interior nodes. Typical questions of interest include:  
what  the likely ancestral state was, whether  
a particular trait evolved just once in the tree or several times at different clearly identified epochs, and 
how reliable our estimates of ancestral states at internal nodes of the tree are. 
 It is this last question that we are concerned with in this paper.  Using both mathematical and simulation-based analyses, we provide new results concerning the performance of various methods for predicting the ancestral states in a tree.
Our work complements and builds on earlier work in this area ({\em c.f. } \citet{fis}, \citet{gas},  \citet{li2}, \citet{li3}, \citet{mad}, \citet{mos}, \citet{zha}, \citet{zha2}) much of which has
focused on the mathematical performance of maximum parsimony, with an emphasis on tree root prediction rather than on the global scenario of all changes along the tree. Two recent papers \citep{roy, sus} have further investigated the relative merits and limitations of  various ancestral state reconstruction methods;  the former notably shows that the performance ranking of likelihood-based methods heavily depends on the tree topology, while the later deals with reconstructing ancestral state frequencies, rather than
the precise state that occurred on every sequence site.

Some of our results apply to all possible methods in placing upper bounds on the reliability of any estimates. However,  we are also interested in comparing the performance of particular methods, such as {\em Majority Rule} (MR), {\em Maximum Parsimony} (MP) and {\em Maximum Likelihood} (ML).  These methods require increasing levels of  knowledge concerning the tree (for MR the tree is irrelevant, for MP we require the tree topology but not the branch lengths, and for ML we require the tree, the branch lengths and a substitution model). Moreover, there are efficient techniques for predicting the scenario of ancestral changes (e.g. for ML, there is the Pupko algorithm \citep{pup}).   We formally demonstrate that any two of these methods can perform very differently on 
the same tree, and that while ML is `best' if the information required is available, MR can either be much better or much worse than MP, depending on the tree structure and the branch lengths. 

A further focus of this paper is the question of whether the root can be predicted with more or less accuracy than the other internal nodes. The question is of interest since
although there are more leaves below the root than below an internal node, the root is also the most ancient node in the tree, and thus is the  `most distant' from the data we observe today.

We first show that there are situations
whether the root can be estimated either much more accurately or much less accurately than the internal nodes.  Turning to Yule pure-birth trees, we establish a result that seems slightly surprising at first:  in certain regimes, where it is impossible to  predict the state at the root of the tree accurately, it is still possible to estimate the ancestral state at a randomly chosen internal node with some accuracy.

The structure of this paper is as follows: we first present basic definitions and an information-theoretic lemma, before mathematically investigating  the relative performance of different methods for predicting  both the root state and the states at internal nodes. We then analyse the  expected performance of methods on Yule trees mathematically.  Simulations confirm this analysis and show that the results still apply with more realistic, non-molecular clock trees and complex models.   We end with a brief discussion.

\section{Mathematical concepts and tools}

\subsection{Definitions}

Consider any method $M$ for predicting ancestral states at the internal nodes of a rooted tree $T$ based on data that consist of the observed state at each leaf $l$ in the leaf set $\L$ of $T$. 
Given any particular assignment of states $x_\L$ to the leaves of $T$,
we will let $M(x_\L, T, v)$ denote the state that method $M$ estimates to be the ancestral state for internal node $v$.  
A standard assumption in statistical phylogenetics is  that the states have evolved on the tree according to some stochastic process (model).
In this case, the assignment of states to the set of leaves (denoted here as $X_\L$) and to any node $v$ of the tree (denoted here as $X(v)$) are random variables. 

We are interested
in the probability that any given method $M$ is able to  predict the state at $v$ correctly given the states at the other leaves, that is $\PP(M(X_\L, T, v) = X(v)).$ 
We call this probability the
 {\em predictive accuracy} of method $M$ at a node $v$ of $T$, and denote it as $PA_M(v, T)$. Thus:
$$PA_M(v, T) := \PP(M(X_\L, T, v) = X(v)).$$

We are also interested in the average of this probability over all the internal nodes of the tree (including the root), and so we let $\overline{PA}_M(T)$ denote the average of $PA_M(v, T)$ over all internal nodes of $T$.   Equivalently, $\overline{PA}_M(T)$ is the predictive accuracy of $M$ at an internal node selected at random.

There is a trivial `lower' bound for these predictive accuracy measures over all methods $M$, namely the one we obtain by the rather naive `method' in which  the leaf data are ignored altogether and 
the state at $v$ is estimated to be the most probable {\em a priori} state under the model (in this case $M(X_\L, T, v)$ is independent of $X_\L$).

\subsection{Models}

In this paper, we deal with time-reversible continuous Markovian processes.  Therefore  if $(\pi_1, \ldots, \pi_r)$ denotes the (unique) stationary distribution on the $r$ states then this trivial lower bound on $PA_M(v,T)$, for any node $v$ of $T$,  is just:
\begin{equation}
\label{pieq}
\pi:=\max_i \{\pi_i\}.
\end{equation}
For instance, for the Jukes-Cantor model or the Kimura 2ST model, we have $\pi = 1/4$.  
Both $PA_M(v, T)$ and $\overline{PA}_M(T)$ lie between $\pi$ and $1$, with $1$ corresponding to perfect prediction, and $\pi$ corresponding to a prediction that is no better than the naive method that ignores the data.
We refer to $\pi$ as the {\em trivial bound} for $PA_M(v,T)$ and $\overline{PA}_M(T)$.

Note, however, that $\pi$  is not a universal lower bound for $PA_M(v,T)$ and $\overline{PA}_M(T)$;  for example,  an even worse `method'  is to systematically predict the state having the lowest {\em a priori} probability.

 In \citep{gas}  we defined a general time-reversible (GTR) process on $r$ states  to be a {\em conservative model} if  the original state is always more probable that any
of the alternative states; formally:
\begin{equation}
\label{ineqpi}
\pi_{ii}(t) >\pi_{ij}(t), \mbox{ for all $t\geq 0$ and all $i \neq j$},
\end{equation}
where $\pi_{ij}(t)$ denotes the transition probability of ending in state $j$ after time $t$ given that $i$ was the initial state.
Eqn. (\ref{ineqpi}) is referred to as the `forward inequality' in \cite{sob}. 

Notice that if $\pi_j$ denotes the equilibrium distribution for any conservative model, then $\pi_{ij}(t)$ converges to $\pi_j$ as $t$ increases and for all initial states $i$.  Therefore,  from Eqn. (\ref{ineqpi}), we have 
$\pi_i \geq \pi_j$ for all $i, j$.   This implies that any conservative model on $r$ states has a uniform equilibrium distribution $(1/r, 1/r, \ldots, 1/r)$, and so (from Eqn. (\ref{pieq})), we have  $\pi= 1/r$.

\subsection{Ancestral state prediction methods}

We consider three main classes of methods for predicting ancestral states, each of which requires a different degree of knowledge concerning the tree. 
The simplest method, and the one that ignores the tree totally, relying just on the states of the leaves, is MR (majority rule). This method estimates the state at a node $v$ as the state
that occurs most frequently among the leaves that lie in the clade that has $v$ as its root (any tie is broken uniformly at random). 
A simple method that takes the tree topology (but not its branch lengths) into account is MP (maximum parsimony) which estimates the state at a node to be the one that minimizes the number of 
substitutions required to explain the evolution of the states observed at the leaves  on the tree.   Finally, if one knows the tree topology, the branch lengths and a substitution model for describing the evolution of the states in the tree, then ML (maximum likelihood) provides a further 
approach to estimating ancestral states \citep{pag}. 

For both MP and ML estimation, it is clear how to estimate the state at the root of the tree. For other internal nodes, there are various options, which we discuss later. For example, it is always possible (and straightforward) to use a root prediction method based on the subtree rooted with $v$. This approach is expected to be less accurate than if we were to  consider all tree leaves when predicting any given node $v$. But mathematical proofs are easier and give a lower bound for the predictive accuracy of more sophisticated approaches (to be described later). 

For the mathematical analysis, we mostly deal with MR and MP.  In  the results that follow, we also describe some other simple
methods for predicting ancestral states in a tree.   Note, however, that some of the likelihood approaches (namely the method that involves predicting the maximum posterior probability ancestral state) can be shown to confer the largest predictive accuracy amongst all methods ({\em c.f.} theorem 3.1 of \citet{stesze}).  This means that our positive results (convergence of the predictive accuracy to 1) obtained with MR or MP  in models that have equal {\em a priori} probabilities of states, also apply to ML methods.

\subsection{Predicting the ancestral states simultaneously at all nodes}

One can  also try to predict the exact history of character evolution on a tree -- in other words, the ancestral state at every internal node. 
In general, this last task is difficult to guarantee with any accuracy, particularly when the tree is large.  It might seem that this problem is hopeless; however, the states at the internal nodes are highly correlated and so the probability of an accurate complete reconstruction for large trees may not be small.    Although few results are available to guarantee the accuracy of reconstructing a complete scenario of changes, there exists a rigorous and explicit lower bound on the accuracy of reconstruction,  under very strong assumptions of very low substitution rates, even for trees with many taxa.

More precisely, suppose the state changes are rare enough that 
any two edges with changes are separated by at least three edges with no changes.  Then the MP reconstruction of the character state changes in the tree is not only unique but it is guaranteed to coincide exactly with evolution of the character within the tree.
This is a combinatorial result, but it translates through to a stochastic bound -- if the probability of transitions are low enough then we are almost sure to be able to reconstruct the transitions within the tree accurately. 

For  details, the reader is referred to proposition 9.5.1 of \citet{ste2005}, which provides an example:  if $n=10,000$, and the probability of a substitution under (say) a Jukes--Cantor model  on each edge is $2\times 10^{-4}$,  we obtain an accuracy for complete reconstruction of $0.99$. This result is a `worst-case' analysis, and in practice, accurate reconstruction may be possible at a higher substitution rate.

\subsection{Information loss}

Information theory provides a useful way to obtain a bound on predictive accuracy that applies across all methods.  
To describe this, first recall that the {\em mutual information} of two random variables $X$ and $Y$ 
is defined by: $$I(X,Y) = \sum_{x,y} \PP(X=x\&Y=y)\ln \left(\frac{\PP(X=x \&Y=y)}{\PP(X=x)\PP(Y=y)}\right).$$ 
$I(X;Y)$ is a non-negative symmetric measure which vanishes precisely
when $X$ and $Y$ are independent, and which has a number of attractive properties \citep{cov}. 
Moreover,  if the mutual information between the states at the leaves and the state
at an internal node is small, then no method can accurately predict the ancestral state at that node from the observed states at the leaves. This is sometimes formalized
in the form of Fano's Lemma \citep{cov}; however, we describe a bound that is more explicit for our purpose,  the proof of which 
is provided in the Appendix.

\begin{lemma}
\label{rab}
For {\em any} ancestral state estimation method $M$ (deterministic or randomized), and any internal node $v$ in a tree $T$, we have:
$$PA_M(v,T) \leq \pi+ \sqrt{ \frac{1}{2}I(X(v); X_\L)}.$$
\end{lemma}

As an application of this lemma, suppose we have a two-state symmetric process (the Neyman two-state model), and we are interested in predicting the state at the root ($\rho$)  from the leaf states
by some method.  Then for a given tree $T$ with $n$ leaves, mathematical results from \citet{eva} can be used to show that:

\begin{equation}
\label{eva2}
I(X(\rho); X_\L) \leq n \exp(-4qt),
\end{equation}
 for an ultrametric tree of height $t$, where $q$ is the substitution rate.
Lemma~\ref{rab} now tells us that  the predictive accuracy of any method to estimate the root state will decay towards the trivial bound unless the number $n$ of leaves grows exponentially with the expected number of
substitutions between the root and any given tip ($qt$); indeed, the number of taxa required cannot be much less than some constant times  $e^{4qt}$.

\section{Mathematical Results I: The range of possibilities}

We have described above an application of Lemma~\ref{rab} that shows that, in a certain regime, all methods must have low predictive accuracy in estimating the root state. 
However it oversimplifies matters to say that this is because the root is simply `too ancient';  there are some trees for which the root can be estimated with higher accuracy than a more
recent node (an example is provided in figure 5b of \citet{sob}). 

Here we carry this a step further, and show that there are trees for which the state at the ancient root node can be predicted with arbitrarily high accuracy, and yet none of the other internal nodes
can have their state predicted with an accuracy much higher than the trivial bound.   We also demonstrate a less surprising converse relationship.
This is summarized in the following result.

\begin{theorem}
\label{firstthm}
\mbox{ }

For any conservative GTR  process, and for $\delta>0$, the following hold: 
\begin{itemize}
\item[(1)]  There are trees and branch lengths for which the  root state can be predicted with an accuracy of at least $1-\delta$ but no method can infer the states at any non-root node with an accuracy that is much better than the trivial bound plus $\delta$.

\item[(2)] There are trees and branch lengths for which the  states at all the non-root internal nodes can be predicted with an accuracy of at least $1-\delta$  but no method can infer the states at the root of the tree with an accuracy that is  better than the trivial bound plus $\delta$.

\end{itemize}
\end{theorem}

{\em Proof:}  A formal proof of Theorem~\ref{firstthm} is provided in the Online Appendix 1; here, we describe  the intuition behind the proof informally.
Part (1) considers the tree shown in Fig.~\ref{tricky_tree1}(a).   Here, the root state can be accurately predicted as $n$ becomes large because,
although each leaf provides little information about the root state when $t$ is  large,  collectively, these leaves (for large enough values of $n$) allow us to infer the root state with an accuracy as close to 1 as we wish. 
By contrast, any other internal node
 has just  three adjacent nodes, each of which will be far from that node for large values of $t$, and in this case, there is no advantage obtained by increasing $n$.  
Notice that the estimation of the root state  is consistent with Lemma~\ref{rab} (and consequential bounds such as (\ref{eva2})), since we are first selecting a large value of $t$ and with this fixed, we let $n$ tend to infinity.

\begin{figure}[ht]
\begin{center}
\includegraphics[scale=.8]{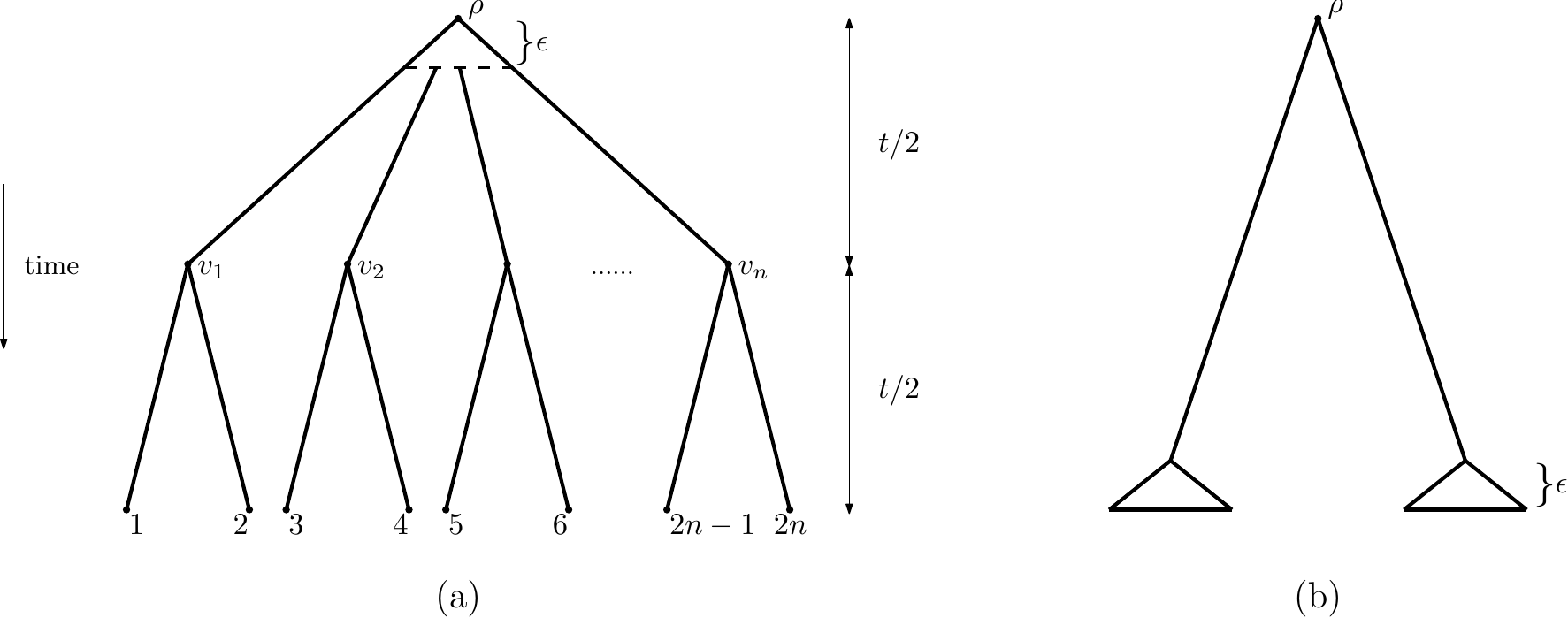}
\end{center}\centering
\caption{(a) A tree for which the state at the root of the tree can be estimated with much higher accuracy than any of the more recent nodes. The top portion of the tree
can be any tree (including a polytomy or a fully resolved tree).  (b) A tree for which  the root state of the tree is difficult to estimate accurately, but the states at the other nodes can be predicted accurately.
}
\label{tricky_tree1}
\end{figure}

The justification of  Part (2) is much simpler when we  consider the tree in Fig.~\ref{tricky_tree1}(b).  If the height ($\epsilon$) of the two subtrees is sufficiently small and the length of the two 
edges that are incident with the root node are sufficiently long, then every internal node can be predicted with high accuracy, while the root node cannot.
In this case, referring again to Lemma~\ref{rab}, for estimating the root state we fix $n$ and $\epsilon$ but let $t$ tend to infinity.

\hfill$\Box$

The accuracy of methods for estimating an internal node (including the root node) can vary considerably between methods -- for some tree shapes,
one method can have high predictive accuracy while another may have a low one;  for a different tree shape, the relative predictive accuracy of these two methods can be reversed.

\begin{theorem}
\label{nextthm}

For any $\delta>0$ (as small as we wish) and a two-state symmetric substitution process:
\begin{itemize}
\item[(1)] there are trees and branch lengths for which  the predictive accuracy of MP for root state estimation is at least $1-\delta$, while that of MR is less than the trivial bound plus $\delta$;
\item[(2)] there are trees and branch lengths for which the predictive accuracy of MR for root state estimation is at least $1-\delta$, while that of MP is less than the trivial bound plus $\delta$;
\item[(3)] assuming the branch lengths of the trees in Part (1) and Part (2) are known, the predictive accuracy of ML is at least $1-\delta$ in both cases. 
\end{itemize} 

\end{theorem}

{\em Proof:}  A formal proof of Theorem~\ref{nextthm} is provided in the Online Appendix 2;  but again we describe  the intuition behind the proof informally.

For Part (1), consider the tree in  Fig.~\ref{tricky_tree2}(a). In this case  MR is dominated by the distribution of states in the leaves of the left-most subtree $T_0$. However, when $L$ is large, the root of this 
subtree ($v$)  is distant from the root node ($\rho$),  so the majority estimate from this subtree will be a poor indicator of the state at this ancestral node $\rho$.  
By contrast, MP performs well if the edges in the two right-hand subtrees $T_1$ and $T_2$ are both sufficiently short and they attach close to the root of the tree
-- for any given large value of $L$ (which, in turn,  fixes the height of the tree), this can always be achieved by making $n$ sufficiently large  and selecting a balanced binary tree for $T_1$ and $T_2$.   In that case, proposition 2.1 of \citet{gas} assures us that the first pass of the Fitch-Hartigan algorithm will be an accurate predictor of the root of $T_1$ and $T_2$, and these accurate estimates effectively determine the predicted state at the root $\rho$ regardless of the possibly unreliable estimate of the root state in the more leaf-rich left-hand tree $T_0$.

For Part (2), consider the tree in Fig.~\ref{tricky_tree2}(b) which has $n$ leaves and all its pendant edges are long, but all its other edges are very short (this might arise, for example, in an early `rapid radiation' scenario).  If this initial radiation is sufficiently short, then each node in the top part of the tree will be in the same state as the root of the tree, and so the states at the leaves represent independent and identically perturbed samples of this state.   MR has a predictive accuracy converging to 1 as $n$ grows (the height of the tree is assumed to be fixed)  since each leaf provides  a (nearly) independent estimate of the root state. However, for MP
the topology of the top part of the tree plays a crucial role -- for instance, if one leaf is incident with the root, then this leaf (which is a poor indicator of the root state by itself) can have a large influence on predictions of the 
root state  under MP.   More generally, we show that if the topology of the top part of the tree is a caterpillar tree  then the predictive accuracy of MP will converge  to a value that is no more than the trivial bound plus $\delta$ as $n$ tends to infinity, provided that (i) the height of the tree is sufficiently large (relative to $\delta$, not $n$) and that the height of the top part of the tree ($\epsilon$) is sufficiently small. 
For full details of these arguments, the reader is referred to the Online Appendix 2.

For Part (3), the two-state symmetric model has a uniform distribution of states at the root so ML estimation maximizes the predictive accuracy and is more accurate than both MP and MR, as noted earlier.   Incidentally, if the branch lengths are not known, this is no longer the case even for the two-state model, since if the branch lengths
are regarded as nuisance parameters in the ML estimation of the root state, then this method becomes identical to MP \citep{tuf}.
\hfill$\Box$

\begin{figure}[ht]
\begin{center}
\includegraphics[scale=1.0]{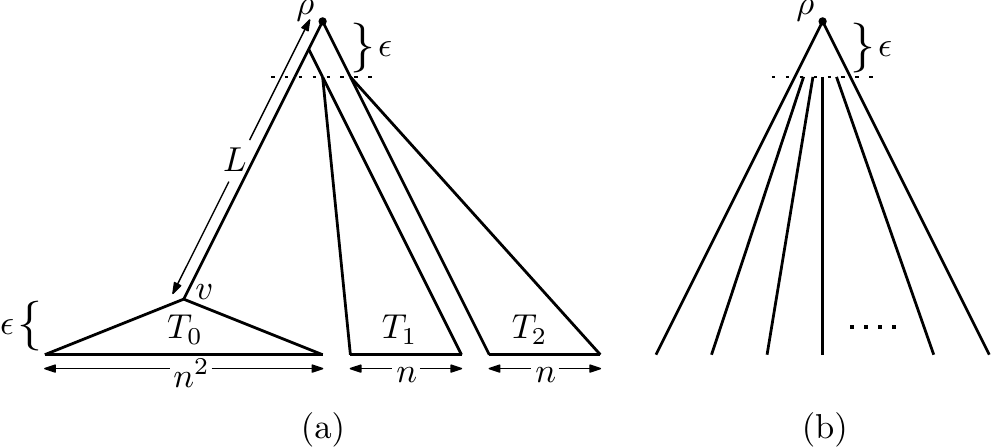}
\end{center}\centering
\caption{The accuracy of root state estimation methods as $n$ grows with the depth of the tree fixed. Portions of the tree not shown in detail are described in the text (a) A tree for which the accuracy of MP converges to 1, but that of MR declines towards the trivial bound;  (b)  a `rapid radiation' scenario where the accuracy of MR converges to 1, but the accuracy of MP converges to a lower value as $n$ grows.
} 
\label{tricky_tree2}
\end{figure}

 Theorem~\ref{nextthm}(2) shows the predictive accuracy of MP  in estimating the root state can be close to  the trivial bound if the top part of the tree has the shape of a caterpillar tree.  However, such trees are highly unbalanced, and so it is pertinent to ask whether the predictive accuracy would improve if the tree was more balanced. We  thus  consider the extreme 
case of  a completely balanced tree (i.e. there are $2^h$ leaves all at depth $h$ from the root), and then make some remarks concerning the case of  random trees.    For the case of a completely balanced tree we have  a quite different predictive accuracy result to the caterpillar  in the following result (the proof of which is in Online Appendix 3).

\begin{proposition}
\label{parsimony_yule}

Consider the tree shown in Fig.~\ref{tricky_tree2}(b). If  the top part of the tree has a completely balanced topology then as $\epsilon$ converges to zero, and $n$ grows (the height of the tree is fixed), the predictive accuracy of MP in estimating the root state in a symmetric two-state model converges to 1.

\end{proposition}

If we were to replace the completely balanced tree with a tree topology selected from the uniform distribution on rooted trees (the so-called `PDA' distribution) then the resulting predictive accuracy of MP has a limit that is strictly less than 1 (since a PDA tree has a positive probability ($\approx 0.5$) of having a leaf adjacent to the root).  Alternatively, replacing the completely balanced tree with a 
Yule-Harding topology appears to lead to similar limiting behaviour as for the completely balanced tree, though we do not have a rigorous proof of this claim.  In any case, the convergence in the Yule-Harding setting is quite slow. For example, with $n=1000$ and $1-p=0.55$, the predictive accuracy of MP for estimating the root state for a Yule-Harding topology can be calculated
exactly (by the recursion described in the Online Appendix 3) and it turns out to be around 0.69; by contrast, for MR, the predictive accuracy is more than 0.99.

\newpage

\section{Mathematical results II: Expected accuracy under Yule trees}

The examples described so far that exhibit the limits of predictability and unpredictability involve trees that are in some sense `extreme' cases.  Thus it is pertinent to ask what we should expect for `typical' phylogenetic trees.  This requires specifying some model for generating a tree and branch lengths, and in evolutionary biology the simplest such model is the Yule pure-birth model.  Despite the simplicity of this model, it nevertheless provides a reasonable approximation to the shape of empirical evolutionary trees \citep{mcp}.
In  this section, we study the predictive accuracy of ancestral state reconstruction in a Yule tree that is grown either for a fixed time $t$ (in which case the number of leaves is a random variable) or is sampled when it has $n$ leaves (in which case the tree height is a random variable).   We are interested in limiting results (i.e., what happens as $t$ or $n$ becomes large?).  The predictive accuracy for the state at the root node, or at a randomly selected internal node is dependent on the ratio of two parameters, the speciation rate (in the Yule model) and the substitution rate.  

 Our main result seems at first somewhat  surprising -- in some regions of parameter space the root state cannot be inferred accurately, yet a randomly selected internal node can be.    Before stating this result more formally, we first need to define precisely what we mean for predictive accuracy when the tree $T$ is randomly generated.
 Let $\T$ denote a random rooted phylogenetic tree generated by the Yule model.  Consider any ancestral state reconstruction method $M$. Then conditional on $\T=T$ 
  the predictive accuracy for estimating the state at the root of the tree is   $PA_M(T, \rho)$.   Thus the  predictive accuracy of $M$ for root state inference of a Yule tree is: 
   $$PA_{M;Y}(\rho) := \EE[PA_M(\T, \rho)]$$ 
  where $\EE$ refers to expectation with respect to the randomly generated Yule tree $\T$.
  
 Similarly, the  predictive accuracy of $M$ for inferring the state at an internal node selected uniformly at random in a Yule tree is denoted $\overline{PA}_M(T)$, and so:
  $$\overline{PA}_{M;Y} := \EE[\overline{PA}_M(\T)].$$ 

We begin by summarizing known results for the special case of the symmetric two-state model. 

\begin{proposition}
\label{secondthm}
\mbox{ } 

Consider a Yule tree grown for time $t$ (from a single lineage), and a two-state symmetric stationary substitution process on the tree.  Then, as $t$ grows:
\begin{itemize} 
\item[(1)]  the predictive accuracy of MP for estimating the root state of the tree converges to the trivial bound ($\frac{1}{2}$) if and only if the
speciation rate is less than {\bf six} times the substitution rate;
\item[(2)] the predictive accuracy of any method for estimating the root state converges to the trivial bound if the speciation rate is less than {\bf four} times the substitution rate.
\end{itemize}
\end{proposition}  

The proof of part (1) is from \citet{gas}.  For Part (2), 
Inequality (\ref{eva2}) applies once we condition on the number of leaves of  the Yule tree. Consequently,  an upper bound on 
the mutual information between the root state and the leaf states is the expected value of $N_t \exp(-4qt)$ where $N_t$ is the number of leaves in a Yule tree grown for time $t$.
Moreover, the expected value of $N_t$ is $\exp(\lambda t)$, where $\lambda$ is the speciation rate.  Thus the predictive accuracy decays exponentially fast to zero if $\lambda < 4q$ (this result was described further in \citet{li}).

We now turn to the main result in this Section, which formalizes the notion that it is easier to predict the state at a randomly-selected node in a Yule tree, than the root state.

\begin{theorem}
\label{next2thm}
\mbox{ }

\begin{itemize}

\item[(1)] With a Yule tree $\T$, and any GTR process,  the accuracy of any method in predicting the state at the root (i.e. $PA_{M;Y}(\rho)$) converges to the trivial bound as $n$ grows (or as $t$ grows), when the mutation rate passes a particular threshold dependent on the speciation rate.

\item[(2)] With a Yule tree $\T$,  and any conservative GTR substitution process there exists a (very simple) method $M$ for which the  accuracy of predicting the state at a randomly selected node (i.e. $\overline{PA}_{M;Y} $) does not converge to the trivial bound as $n$ (or $t$) grows,  for any fixed mutation rate.
\end{itemize}

\end{theorem}

The formal proof of this theorem is provided in the Appendix, but here we offer some informal comments as to the underlying intuition behind the claims. Part (1) is similar to the statement of part (1) of the Proposition~\ref{secondthm} but it differs from it in  an important way: the previous theorem was restricted to the two-state symmetric model, while here we are dealing with more general processes (the price we pay for the extra generality is that the bounds obtained are weaker). 

Turning to Part (2), a key observation is that in  any rooted binary tree at least half of the internal nodes are adjacent to at least one leaf, and the lengths of the pendant edges they are incident with have expected length of
$1/2\lambda$ on average \citep{moo}.   Thus, for at least half the internal nodes that are adjacent to a leaf, selecting the state of that leaf as an estimate leads
to  non-vanishing predictive accuracy, regardless of the speciation and substitution rates, and $n$. 
\hfill$\Box$

The simplicity of this method makes the proof easy; more sophisticated, realistic approaches are more accurate, typically ML with uniform priors or maximum posterior probability estimation.  In other words, the (positive) result in Part (2) stands for a large variety of methods, while the (negative) result in Part (1) applies for all possible methods.

\section{Simulation results}

Theorem~\ref{next2thm} states that for any method, the predictive accuracy of reconstructing the tree root vanishes as $n$ (or $t$) grows, when the speciation/mutation rate ratio passes below a particular threshold. With the two-state symmetric model, this threshold is equal to 4 (for any reconstruction method, e.g. ML), while it is equal to 6 for parsimony. 
Moreover, Theorem  \ref{next2thm} shows that the accuracy of a very simple method in predicting a randomly selected internal node does not vanish as $n$ (or $t$) grows, when this rate ratio is fixed. Due to the simplicity of the reconstruction method used in this proof, this result still holds for more sophisticated approaches, such as those  based on parsimony or likelihood. However, Theorem ~\ref{next2thm} does not provide any quantification. We do not know how quickly the ability to reconstruct the tree root vanishes or the  extent to which the internal nodes can be reconstructed. This section uses computer simulations to answer these questions. We first simulate Yule trees with variable numbers of tips and speciation/mutation rate ratios, and assess the accuracy of a number of reconstruction methods based on majority, parsimony and likelihood. The results show that reconstructing internal nodes is indeed much easier than reconstructing the tree root. In a second series of simulations, we show that these results still holds when using more realistic, non-molecular clock trees and a standard substitution model for DNA sequences.

\subsection{Yule trees and the two-state symmetric model}

We generated Yule trees with $n = 10, 100$ and 1,000 tips. Next,  binary (0/1) sequences of length 50 were randomly generated and evolved along the tree using the symmetric Neyman model. The speciation/substitution rate ratio was equal to 1, 2, 3, 4, 5, 6, 8, 12 and 20, thus having a focus on the 4--6 region where the accuracy of the various methods is expected to drop and be clearly different from one method to another.

We compared a number of reconstruction methods:
\begin{itemize}
\item
{\bf Majority}: For any given node $v$, we select the majority state among $v$'s  descendants; in the case of a tie, we randomly select 0 or 1, with equal probability. This method uses partial information when predicting a non-root node $v$, as only the descendants of $v$ are accounted for. Moreover, the predictions are made independently for each of the tree nodes.

\item {\bf Parsimony}: We have the choice among several options:
\begin{itemize}
\item
Parsimony-Down: just as with Majority, we only look into the subtree rooted with $v$, using the standard Fitch-Hartigan algorithm \citep{fit, har}. Due to partial information and independence among node predictions, we do not expect high accuracy using Parsimony-Down.
\item
Parsimony-Acctran and Parsimony-Deltran: we use now all tips to predict the ancestral state of any given internal node, and thus expect better results than with Parsimony-Down (and Majority). Acctran and Deltran are two heuristic procedures \citep{mad2, swo}. Both select one most parsimonious global change scenario (among many, for most data-sets). Acctran means ``accelerated transformations" and favours mutations close to the tree root, thus tending to avoid convergent mutations, but accepting reverse mutations. Deltran means ``delayed transformations" and favours mutations close to the tips, thus preferring convergent mutations rather than reverse mutations. Acctran is typically used with morphological characters (convergent evolution is then unlikely), while Deltran is often used with geographic annotations (ÒconvergentÓ evolution then means multiple introductions into some country or region and is quite possible; see e.g.  \citet{wal}).
\item
Parsimony-Independent: we compute all most parsimonious state assignments for all nodes, and then select one of the two states (independently and with equal probability $0.5$)  for every ambiguous node. Computations are performed using the DownPass algorithm that is described in the MAC Clade user guide \citep{mad2}. DownPass is faster but equivalent to re-rooting the tree on every node $v$ and running the Fitch--Hartigan procedure, with a slightly modified last step because $v$ has now three root descendants. With this procedure, the selected scenario may not be one of the most parsimonious scenarios, and it is easy to see with examples that it may be not parsimonious at all. This actually occurs for most of the data sets generated in this study. We tested this (non-standard) procedure to assess the importance of the dependence among predictions and to evaluate  ML approaches, where we have a similar choice between using the most likely joint scenario or the marginal likelihood of each of the nodes (see below).
\end{itemize}
\item
{\bf Maximum likelihood}: We have similar options as we have with parsimony:
\begin{itemize}

\item
Likelihood-Down: just as with Majority and Parsimony-Down we only look into the subtree rooted with $v$ using the standard pruning algorithm \citep{fel}, and then select the most likely state. We thus use partial information and node predictions are independent (actually the dependence is not explicit, contrary to the following option).

\item
Likelihood-Best: we use the dynamic programming algorithm proposed by \citet{pup} to infer the most likely joint change scenario.
\item
Likelihood-Marginal: to (independently) select the most likely state for every node $v$, we use the marginal probabilities obtained using the pruning algorithm after re-rooting the tree on $v$ (with an easy adaptation for the last three-descendant step).

\end{itemize}

\end{itemize}

We measured the accuracy of all these methods in reconstructing the ancestral state at: (1) the tree root and (2) any randomly selected internal node (including the tree root).   
These accuracy measures are simply the proportions of correct root and node predictions in the simulated data.   We also measured (3) the method accuracy
in reconstructing the changes that occurred along the tree branches. Let $e$ be a randomly selected branch, and let $i$ and $j$ be the states observed during simulations at both ends of $e$ for a given site; the change between $i$ and $j$ is correctly reconstructed when the studied method correctly reconstructs both $i$ and $j$ (note that $i$ may be equal to $j$).  With pendant branches, the leaf state is known, and thus only interior branches are accounted for in this measure.

This `branch accuracy'  measure is used to assess and compare the performance of independent/dependent methods.   Let us assume that node prediction has a success probability of $p$ and that it is independent from one node to another (as we basically expect with Parsimony-Independent and Likelihood-Marginal), then the  success of reconstructing the changes along the internal branches should be equal to $p^2$.  If we now assume that the prediction successes are fully correlated at both  ends of $e$, that is, both predictions are simultaneously correct or wrong, then the expected success of predicting the change along $e$ is equal to $p$; for example, with $p$=0.7, we expect values of  0.49 and 0.7, respectively with independent and perfectly concerted predictions, that is, very different accuracies.

The three accuracy measures (root, node and branch) are averaged over the 500 trees, 50 sites, and $(n-1)$ internal nodes or $(n-2)$ internal branches, for each of the speciation/substitution rate ratios and $n$ values.  The results are displayed in Figure~\ref{simulate} and in the Online Appendix 4. The main findings are as follows:

\begin{figure}[!ht]
\begin{center}
\includegraphics[scale=.8]{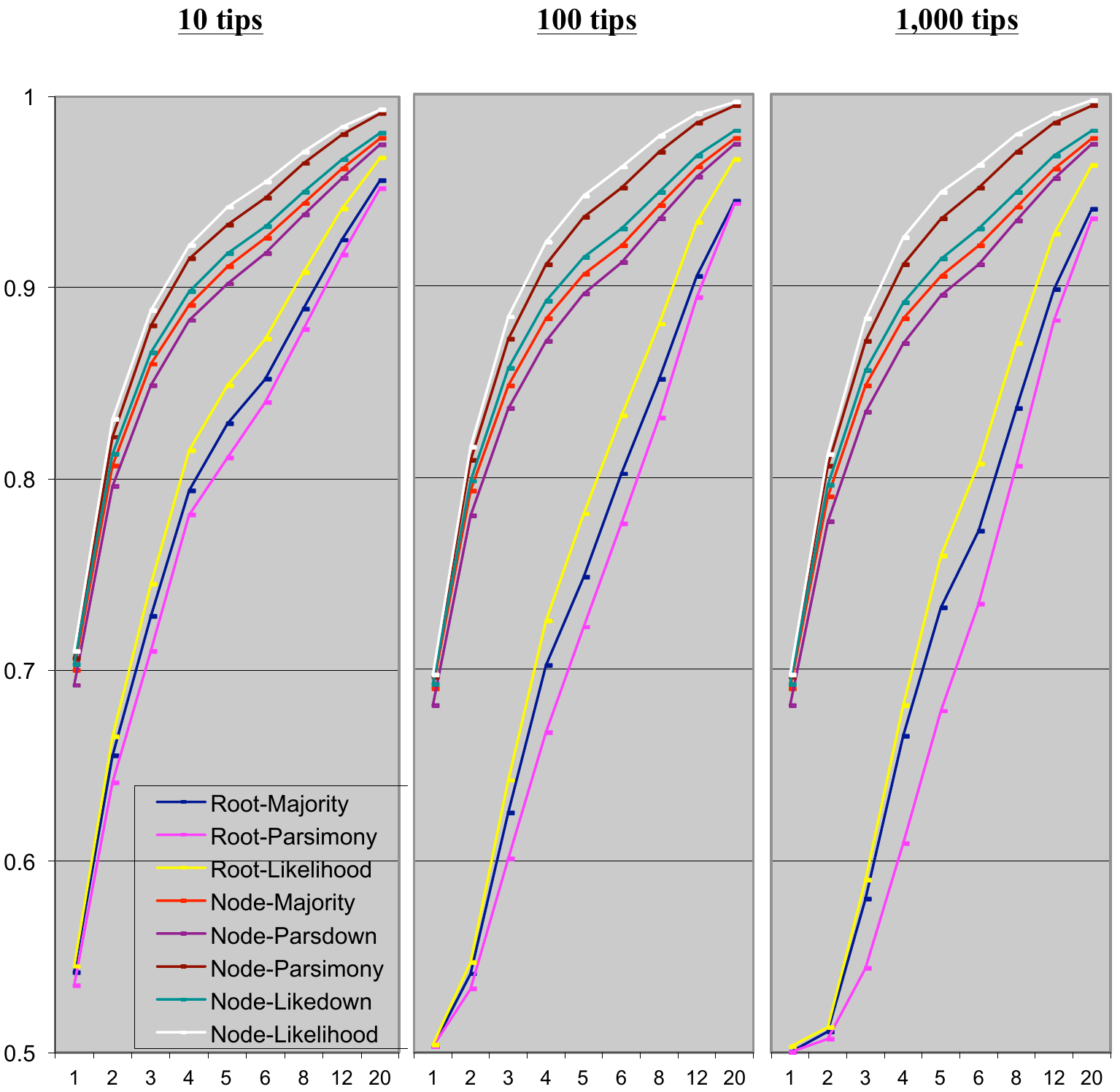}
\end{center}\centering
\caption{A comparison of the accuracy of various methods for reconstructing the root state (denoted ÒRootÓ in the legend) and any randomly selected interior node (denoted ÒNodeÓ). Majority, Parsdown and Likedown only use the descendants of the node to be predicted, while Likelihood and Parsimony use all tips in their predictions (this distinction does not hold when predicting the tree root). All parsimony methods (except Parsimony-Down, referred to as Parsdown for short) have very similar results, so here we use  Parsimony-Acctran in Root- and Node-Parsimony. Similarly, we use Likelihood-Best \citep{pup} in Root- and Node-Likelihood. X-axis: speciation/substitution rate ratio; Y-axis: proportion of correct predictions; 500 trees with sequences of length 50 were used for each of the conditions (number of tips and rate ratio).} 
\label{simulate}
\end{figure}

\begin{itemize}
\item
Regarding root prediction, the results are congruent with \citet{gas}. Likelihood is best, as expected. Majority is better than Parsimony and is surprisingly accurate, despite its simplicity and the fact that it does not use the tree topology. However, these results should not be overemphasized. In many cases (e.g. with morphological or geographical characters), we may have some sampling bias in the number of occurrences of some of the character states, in which case Majority is expected to perform poorly. For example, when half of the tips with a given character state (say 0) are not sampled, while all other tips (those having 1) are sampled, the accuracy of all methods drops but Parsimony becomes better than Majority; with 100 tips and a speciation/mutation rate ratio of 6, the accuracy of reconstructing the tree root is 0.70, 0.74 and 0.79 for Majority, Parsimony and Likelihood, respectively, against 0.80, 0.77 and 0.83 with no sampling bias. With this (strong) sampling bias, Majority is thus more affected  than Parsimony and Likelihood, which both appear remarkably robust.

\item
As expected from Theorem~\ref{next2thm}(1), the performance of all methods in predicting the tree root drops down when the number of tips increases. For example, with a rate ratio of 4, the accuracy is $\sim$0.8, $\sim$0.7 and $\sim$0.65 with $n = 10$, $n=100$ and $n=1000$, respectively. Moreover, with $n = 1000$ and a rate ratio of 1 and 2, the accuracy of all methods is nearly the same as that of random predictions, while with a rate ratio of 5, where Parsimony is expected to be poor (see above), we see a clear gap between this method and the two others. Actually, based on these simulations, it seems that the transition value for Majority should be less than 6 (as it is for Parsimony) and may possibly be 4 (the best possible value, attained by Maximum Likelihood). 
\item
However, this decrease in the accuracy of predicting the tree root is rather slow; for example, with a rate ratio of 3, where all methods should become analogous to random guessing (accuracy = 0.5) with large $n$, we still see a clear signal with $n = 1000$. Moreover, with large rate ratios, the decrease is even slower; for example, with a rate ratio of 20, we see very little difference between $n = 10$ and 1000.
\item
If we now compare the accuracy of reconstructing interior nodes, the performance of methods that use all the tips is nearly the same, disregarding whether the predictions are done simultaneously or independently. Specifically, Parsimony-Acctran, Parsimony-Deltran and Parsimony-Independent have nearly identical accuracy, and the same holds for Likelihood-Best and Likelihood-Marginal (see Online Appendix 4). This somewhat surprising finding does not come from the simulation protocol. For example, with Parsimony, there is a large number of nodes (up to 30\% with high mutation rates) with an ambiguous ancestral annotation, meaning that Parsimony-Independent produces highly suboptimal scenarios (but relatively accurate node predictions!). Moreover, when examining the accuracy in reconstructing the changes along the tree branches, we see (Online Appendix 4) that independent and dependent methods perform nearly identically. Actually, we see a slight advantage (of at  most 2\%) for Parsimony-Acctran and Parsimony-Deltran over Parsimony-Independent, while Likelihood-Best and Likelihood-Marginal are almost undistinguishable (the former is possibly a bit better than the latter with 10 tips and the opposite may hold with 1000 tips). Moreover, the accuracy value of both independent and dependent approaches is just slightly above $p^2$, where $p$ is the node accuracy, meaning that all of these methods perform nearly the same as if they were achieving independent predictions from one node to another. This is obviously not true for dependent methods like the approach of \citet{pup},  but their accuracy has not been improved so far (their interpretability likely is). Thus, in the following and in Figure~\ref{simulate}, we present and discuss the results of Parsimony-Acctran (the most standard parsimony option) and Likelihood-Best (the most rational option with maximum likelihood in this context). For the sake of conciseness, these methods are simply referred to as Parsimony and Likelihood.
\item
The two best methods for predicting interior nodes are Likelihood and Parsimony, in this order, and the difference between both is surprisingly small ($\sim$1\% or less in all conditions). This contrasts with root prediction, where the gap is much higher (up to $\sim$8\%). Methods using descendant tips only are clearly behind, with Likelihood-Down being the best of these, Parsimony-Down the worst and Majority in between, as expected from the results on root prediction.
\item
The accuracy of all methods in predicting interior nodes is remarkably similar regardless of the number of tips. It is difficult to see any difference between $n = 100$ and 1000, and the results with $n$=10 are neither worse nor better than with $n = 100$ and $n=1000$. This finding is most likely to be explained by the fact that in a Yule tree, the subtrees are also Yule trees, meaning that in large Yule trees, most of the nodes are contained in small Yule trees with only a few tips. We thus observe very fast convergence of the node prediction accuracy for all methods and conditions, which contrasts with the slow degradation of performance when reconstructing the tree root, which we discussed above (especially for high rate ratios).
\item
The main result from these simulations is that there is a large gap in accuracy when predicting the tree root and interior nodes, especially with low rate ratios where root prediction accuracy vanishes for all methods. For example, continuing above example with a rate ratio of 4, the accuracy in reconstructing interior nodes is $\sim$0.9 for all methods and  values of $n$. In other words, reconstructing the interior nodes is (relatively) easy, while reconstructing the tree root is difficult most of the time and is just impossible for large values of $n$ and small rate ratios. We will show in the next section that this statement still holds when using more realistic, non-molecular trees and a standard DNA substitution model 
\end{itemize}


\subsection{Non-molecular-clock trees and the HKY+$\Gamma$ substitution model}

We reuse here the simulation protocol of \citet{gas}, where we compared the accuracy of various methods in predicting the root state. Similar simulations were used to benchmark the topological accuracy of a large variety of tree building programs \citep{gui,des}, and  their features and parameterizations may be seen as biologically realistic. Here we summarize the main components of this protocol; additional explanations and justifications can be found in the previously mentioned references.

We first generated a Yule tree with $n = 25, 50, 100, 200$ and 400 leaves. This molecular-clock tree was then perturbed by multiplying every branch length (independently) by $(1 + X)$, where $X$ was an exponential variable with parameter 0.5. The factor $(1+X)$ was used (as opposed to, say, $X$) to avoid an excessive number of very small branches. The observed departure from the molecular clock, as measured by the ratio between the longest and shortest root-to-leaf lineages, was equal to $\sim$3.5 on average, a value that is typical in published phylogenies. Finally, the whole tree was rescaled so that the average root-to-leaf distance was uniformly distributed between 0.1 (relatively low divergence) and 1.0 (high divergence). In the previous set of simulations, the height of the tree was increasing with $n$, but here we assume that the average height is kept constant, while we increase the taxon sampling density. We thus expect that the larger the $n$, the more accurate the various methods will be in reconstructing the ancestral root and interior node states. We generated 500 trees using this procedure for each tree size $n$.

DNA-like sequences of 100 sites were evolved along these trees using the HKY model \citep{has} with a transition/transversion rate ratio ($\kappa$) of 4.0 (the default value in most software) and the equilibrium frequencies of A, C, G and T being equal to 0.15, 0.35, 0.35 and 0.15. This HKY model was combined with a discrete gamma distribution of parameter 1.0 with six rate categories. We thus obtained 500 data-sets of 100 sites for each tree size $n$.

Tested methods were essentially the same as in previous simulation study.  However, the \citet{pup} algorithm (Likelihood-Best) is not able to cope with the gamma model of site rates and was not used here. Moreover, the Maximum Likelihood approach uses a number of parameters (branch lengths, nucleotide frequencies, transition/transversion rate ratio, gamma distribution) which are not used by other approaches (Majority and Parsimony) and are generally not (or approximately) known in practical cases. Thus, we used the Maximum Likelihood approach under three settings: 
\begin{itemize}
\item
Likelihood-Down: just as in previous simulations, this method uses only the descendants of the node to be predicted, combined with  HKY+$\Gamma$6 and a complete knowledge of the model parameters.
\item
Likelihood-Marginal: again, we use the same approach as in previous simulations, based on the marginal distribution of state probabilities for every tree node, combined with HKY+$\Gamma$6 and a complete knowledge of the model parameters. This approach somehow provides the ``best possible'' result that can be obtained with our data sets.
\item
Likelihood-Marginal-JC: we used again the same marginal probabilities, but assumed the simple Jukes--Cantor (JC) model which ignores the differences in nucleotide frequencies and transition/transversion rates which were employed to generate the data. Moreover, as this JC model was used without a gamma distribution of site rates, the branch lengths were only approximate. This approach is thus expected to provide a more realistic view of maximum likelihood performance, compared to the performance of  parsimony that is based on similar simplifying assumptions, but does not use (even only approximate) branch lengths.
\end{itemize}

The same accuracy criteria were used as in previous experiments. All results are displayed in Figure \ref{fig_4} and in the Online Appendix 5.  The main findings are as follows (some are similar to what we already observed in previous simulations and are merely  summarized):

\begin{figure}[!ht]
\begin{center}
\includegraphics[scale=0.8]{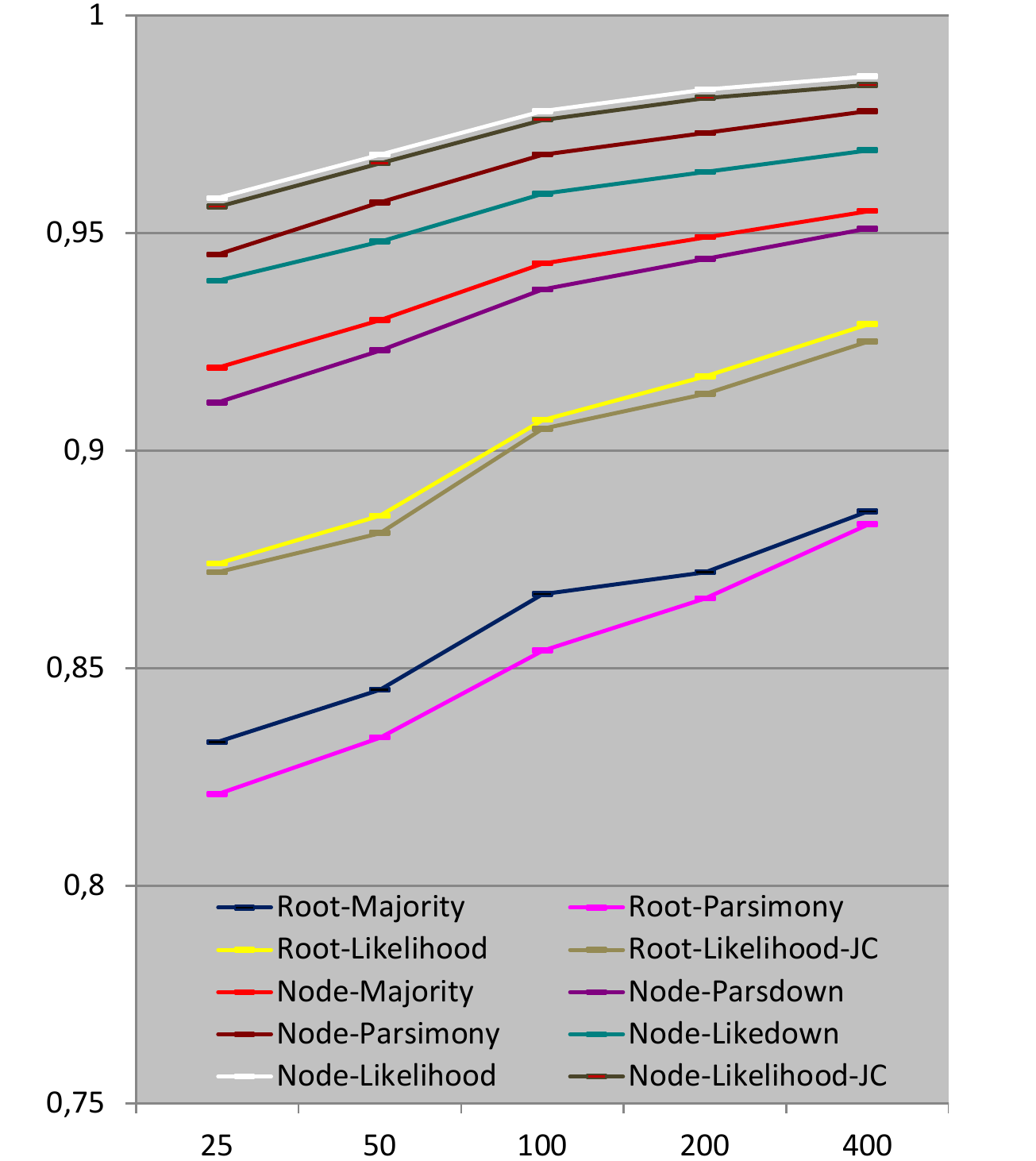}
\end{center}\centering
\caption{
Accuracy of the methods with non-molecular-clock trees and a HKY+$\Gamma$ substitution model.
Note: Root- and Node-Likelihood results are obtained using Likelihood-Marginal with HKY+$\Gamma$6. Root- and Node-Likelihood-JC are based on the same method but use JC model instead. Node- and Root-Parsimony use the Acctran algorithm. Majority, Parsdown and Likedown only use the descendants of the node to be predicted; the latter is launched with HKY+$\Gamma$6. X-axis: $n$, number of tips; Y-axis: proportion of correct predictions; 500 trees with sequences of length 100 were used for each value of $n$.}
 \label{fig_4}
\end{figure}

\begin{itemize}

\item
Majority performs better than Parsimony-Down, both in reconstructing the root state and the  interior nodes.
\item
The three all-tips parsimony methods (Acctran, Deltran and Independent) perform nearly identically, with a slight (probably  non-significant) advantage for Acctran, the results of which are displayed in Figure \ref{fig_4}.
\item
To predict the state of interior nodes, methods using all tips are clearly better than methods using node descendants only; for example, Parsimony-Acctran is better than Likelihood-Down, despite the fact that the latter uses a complete knowledge of the substitution model (branch lengths, nucleotide frequencies, etc).
\item
Again, we do not see any significant difference between methods performing dependent and independent predictions (see the results of Acctran, Deltran and Parsimony-Independent in the Online Appendix 5). Moreover, for all methods, the accuracy in reconstructing the changes along the tree branches is roughly equal to the square of the node accuracy, just as if predictions were made independently at both branch extremities (Online Appendix 5).
\item
Again, we see a large gap between predicting the root and predicting the interior node states. Notably, the gain with Parsimony is 10\% or more.  Despite its simplicity, Parsimony appears to be quite accurate at predicting the interior node states and the changes along the tree branches.
\item
As expected, all method accuracies increase with $n$ (representing taxon sampling density) but the amelioration is relatively slow ($\sim$5\% for all methods when predicting the root state) compared to  augmentation of the taxon number (from 25 to 400).
\item
Lastly, the most surprising finding is the remarkable performance of maximum likelihood when used with the over-simplistic JC model.  The results are nearly the same as with HKY+$\Gamma$6  for both the tree root and interior nodes. Compared to Parsimony, this Maximum Likelihood approach is clearly better at predicting the root state, thanks to the use of (approximate) branch lengths, while Parsimony uses the tree topology only. This robustness is quite encouraging regarding the use of Maximum Likelihood with real data, and can be related to the robustness we have already observed with biased sampling (see above), and the apparent robustness of Maximum Likelihood methods regarding topological errors reported by \citet{han}.

\end{itemize}

\section{Concluding comments}

Ancestral state reconstruction based on a phylogenetic tree allows biologists to estimate where and when important  innovations (the gain, loss or change of some character sate) may have occurred in the evolutionary history of a set of taxa.  Such approaches are also useful in phylogeography (e.g. \cite{sla}; \cite{wal}) to study the evolution of epidemics and their movement and exchanges moving from one country to another.

In this paper we have applied mathematical methods, combined with simulations to quantify how reliable such predictions are likely to be under a variety of models of sequence evolution.  We showed that  predicting the state at a given node that is deep in the tree can be provably arbitrarily close to the trivial bound, even though for a randomly selected node in a Yule tree the accuracy is always lies a separated distance above the trivial bound. 

Yet, this story also has some twists  -- there are trees for which the root can be predicted more accurately than any other node. And Majority Rule, which ignores the tree structure in estimating the root state, has an accuracy that is not much worse than other methods. This may go some way towards explaining the apparent  robustness of ancestral state prediction to the choice of tree, by \cite{han} who stated that ``incorporating phylogenetic uncertainty very rarely changes the inferred ancestral state and does not improve the accuracy of the reconstructed ancestral sequence."
Moreover the prediction methods are relatively robust against other factors such as model mispecification, sampling bias and approximate parameter and branch length values.

Among the three main methods we consider, ML, MP and MR, the first method (ML)  is the most accurate, but it also requires knowing the most about the tree and model. Regarding MP and MR we showed that the relative performance of each depends very much on the tree and its branch lengths, and that neither is universally better than the other. However on Yule trees MR tends to be slightly more accurate than MP  for estimating the root state, but less accurate at estimating randomly selected nodes in the tree than certain versions  of parsimony which account for all tree tips.

For future works, a fundamental mathematical question is to determine, for the 2-state symmetric model, whether the transition value of the speciation-to-substitution ratio for the asymptotic accuracy of Majority Rule is at the lowest possible value 4 or some value higher than 4 (as is the case for Parsimony, which has its transition at the ratio 6). Moreover, it would be interesting to further explore the prediction of the edges on which particular state changes occurred, and the accuracy of complete ancestral reconstruction in the tree. Regarding this last task, the only theoretical result known (described earlier) requires very strong assumptions, and is far from optimal. Moreover, our simulation results indicate that the dependence between predictions made at two edge extremities is not well accounted for by current approaches, which could be a route to design more accurate methods aimed at predicting the changes that occurred along the tree.

\section{Acknowledgments}  MS thanks the {\em Allan Wilson Centre for Molecular Ecology and Evolution} for helping fund this research.  OG was supported by the ANR project PhyloSpace and the PIA Institut de Biologie Computationnelle.


\bibliographystyle{sysbio}
\bibliography{Ancestral_Gascuel_Steel}

\newpage
\section{Appendix: Mathematical proofs (see also the Online Appendix)}

\subsection{Proof of Lemma~\ref{rab}}
\begin{proof}
Firstly observe that for any two discrete random variables $X$ and $Y$, the mutual information $I(X;Y)$ equals the Kullback-Leibler divergence $D(P_{X,Y}||P_XP_Y)$ between $P_{X,Y}$ (the joint probability distribution of $X$ and $Y$) and
$P_XP_Y$ (the corresponding product of marginal probability distribution, and which thereforetreats $X$ and $Y$ as independent). Now,  
Pinsker's inequality ({\em c.f. }\citet{cov}) states that for any two probability distributions $P$ and $Q$ we have:
$$\sqrt{\frac{1}{2}D(P||Q)} \geq \sup_A\{|P(A)-Q(A)|\},$$
where the supremum is over all events $A$. Combing this with the previous observation gives:
\begin{equation}
\label{sumar}
\sqrt{\frac{1}{2} I(X;Y)} \geq \sup_A\{|P_{X,Y}(A)-P_XP_Y(A)|\},
\end{equation}
Now, suppose that $M$ is any (deterministic or randomized) method for  predicting $X$ from $Y$, and let $E$ be the event that $M(Y)=X$
(note, we allow $M$ to be `random' as some methods result in ties that then are generally broken randomly).

Then, from (\ref{sumar}), we have: $$\sqrt{\frac{1}{2} I(X;Y)} \geq \PP_{X,Y}(E)-P_XP_Y(E),$$
and so, equivalently:
\begin{equation}
\label{basicineq}
\PP_{X,Y}(M(Y)=X) \leq \sqrt{\frac{1}{2} I(X;Y)}  +P_XP_Y(E).
\end{equation}
Now, $$P_XP_Y(E) = \sum_{x,y}P_XP_Y(X=x, Y=y)P(M(y)=x) = \sum_{x,y}P(X=x)P(Y=y)P(M(y)=x).$$
Thus if we let $a_x = P(X=x), b_y = P(Y=y)$ and $c_{xy} = P(M(y)=x)$,
then
\begin{equation}
\label{PXPY}
P_XP_Y(E) = \sum_{x,y} a_x b_yc_{xy}.
\end{equation}
Setting $A = \max_x \{a_x\}$,   Eqn. (\ref{PXPY}) now gives:
$$P_XP_Y(E) \leq A \sum_{x,y} b_y c_{xy} = A \sum_y b_y \sum_x c_{xy} = A \sum_y b_y = A,$$
where the second equality relies on the fact that for each $y$  we have $\sum_x c_{xy}=1$.

Thus,  since $A = \max_{x}P(X=x) = \pi$ we obtain $P_XP_Y(E) \leq \pi$, and  
applying this to (\ref{basicineq}) in the case where  $X=X(v$), $Y=X_\L$  we obtain the claimed inequality.

\end{proof}

\subsection{Proof of Theorem~\ref{next2thm}}

{\em Proof of Part (1):}   We exploit a result from \citet{moste} (Eqn. (14.8)) which says that for \underline{any}  method $M$ if we take $v=\rho$ (the root of a fixed tree $T$) then
\begin{equation}
\label{limiting2}
PA_M(\rho, T) \leq \pi+   \sum_{v \in \L} \exp(-qt(v)),
\end{equation}
where $t(v)$ is the sum of the branch lengths from the root to leaf $v$ and $q = \sum_j \min_{i \neq j} q_{ij}$, where $q_{ij}$ is the transition rate from state $i$ to state $j$.

We first consider what happens if we allow $t$ to grow (in which case the number of leaves at time $t$, $N_t$, is a random variable) as this is simpler than allowing $n$ to grow.
For a Yule tree of depth $t$ we have $t(v)=t$ for each leaf $v$. Thus, conditional on $N_t = n$, Eqn. (\ref{limiting2}) gives
$$\PP(M(X_\L, \rho) = X(\rho)|N_t = n) -\pi \leq n\exp(-qt),$$
and since $PA_{M;Y}(\rho)$ is the expected value of this quantity (with respect to the Yule model) we have:
$$PA_{M;Y}(\rho, T) -\pi  \leq \EE[N_t] \exp(-qt).$$
Now $\EE[N_t] = \exp(\lambda t)$ and so $PA_M(\rho, T) \leq \exp((\lambda - q)t)$, which converges to zero, exponentially fast, if $\lambda< q$.

We now consider what happens if we allow $n$ to grow, in which case the time $\tau_n$ until the Yule tree has $n$ leaves is a random variable.  
Firstly, observe that, for such a tree, $t(v)$ takes the same value for all leaves $v$ (since Yule trees are ultrametric).  Conditional  on $\tau_n = t$ we have, from 
(Eqn. (\ref{limiting2}):
$$\sum_{v \in \L} \exp(-qt(v)) = n\exp(-qt). $$

Thus,
\begin{equation}
\label{limiting}
PA_{M;Y}(\rho)  -\pi \leq  n \EE[\exp(-q\tau_n)].
\end{equation}

Now it is a classic result \citep{ken} that the number $N_t$ of leaves in a Yule tree generated for time $t$ and speciation rate $\lambda$ (starting with a single lineage) has a geometric distribution with parameter $1-\exp(-\lambda t)$. In other words, we have: $$\PP(N_t \geq n) = (1-\exp(-\lambda t))^{n-1}.$$
Moreover, by definition we have:  $$\PP(\tau_n \leq  t) = \PP(N_t \geq n), \mbox{ } n \geq 1$$
and so, taking $t= t_n=(1-\delta) \frac{\ln(n)}{\lambda}$  where $\delta > 0$ is chosen so that  $q > \lambda/(1-\delta)$ we have:

\begin{equation}
\label{pptau}
\PP(\tau_n \leq  t_n) \sim \exp(-n^{\delta})
\end{equation}
and so the right-hand side of Eqn. (\ref{limiting}) (namely $n \EE[\exp(-q\tau_n)] $) is equal to:
$$\EE[ne^{-q\tau_n}|\tau_n\leq t_n]\cdot \PP(\tau_n \leq t_n) + \EE[ne^{-q\tau_n}|\tau_n> t_n]\cdot \PP(\tau_n>t_n),$$
which, using Eqn. (\ref{pptau}) we can further bound as follows:
$$ n\cdot \exp(-n^{\delta}) + n\exp(-qt_n)\cdot (1-\exp(-n^{\delta})).$$
The first term in this last expression converges to zero as $n$ tends to infinity, while the second term is bounded above by 
$$n\exp(-qt_n) = n\exp(-q(1-\delta)\ln(n)/\lambda) = n^{1-q(1-\delta)/\lambda},$$ which also converges to zero as $n \rightarrow \infty$  since $q>\lambda$ and $\delta>0$ has been chosen so that  $q > \lambda/(1-\delta)$.

Thus provided that $q> \lambda$ we may select $\delta$ sufficiently small (but positive) to ensure that the predictive accuracy 
of any method $M$ on a Yule tree converges to the trivial bound as $n$ tends to infinity.

{\em Proof of (2):}  For any time-reversible continuous-time Markov process on $r$ discrete states, the transition probability of being in the starting state after any given time $t$ can be written as:
\begin{equation}
\label{aldy}
p_{ii}(t) = \pi_i + \sum_{j=2}^r a_j \exp(-b_i t),
\end{equation}
where  $-b_2, \ldots, -b_r $ are all strictly negative, and comprise the non-zero eigenvectors of the rate matrix, and where $a_j>0$ for all $j$ \citep{ald}.
 We will assume the $b_i$ are ordered in increasing order of absolute value. 
Note that the constants $b_i$ are proportional to the rate of the process (in our setting the substitution rate), but the $a_j$ are independent of this rate (and also of $t$). 
 Eqn. (\ref{aldy}), and the condition $p_{ii}(0) = 1$ implies that:
 $$p_{ii}(t) \geq  \pi_i + (1-\pi_i)  \exp(-b_r t).$$

For a conservative process we have $\pi_i = \pi (= \frac{1}{r})$ and so 
 \begin{equation}
\label{aldy2}
 p_{ii}(t) \geq  \pi + (1-\pi) \exp(-b_r t)
\end{equation}

Now let $T$ be a  binary tree, rooted on an ancestral node of degree 1 (we assume this rather than a root of degree 2 since we wish to model a Yule process grown from a single initial lineage).         
Thus if $T$ has $n$ extant leaves, $T$ has $n$ pendant edges and $n$ ancestral nodes (including the root node). 

Consider the following method $M$ for predicting the state at an ancestral node $v$ of $T$.  Given $v$ select  any  leaf that is in the clade below $v$, and estimate the state at $v$ by the state at this leaf. 
Then, conditional on the evolutionary time from $v$ to the present being $t$, inequality (\ref{aldy2}) gives:
$$PA_M(v, T)   \geq  \pi + (1-\pi)  \exp(-b_r t).$$

The following randomized scheme selects uniformly at random one of the ancestral nodes of the tree $T$.    First select uniformly at random one of the pendant edges of this tree -- call this randomly sampled edge $e$.   If the ancestral node $v$ of $T$ that is incident with $e$ is not adjacent with another extant leaf then select $v$.
Otherwise (i.e. if $v$ is adjacent to two extant leaves) toss a fair coin, and if the outcome is `heads' select $v$, while if the outcome is `tails'  select uniformly at random an ancestral node of the tree that is not adjacent to any extant leaf. 
It can be checked this randomization scheme selects each of the $n$ ancestral nodes of $T$ with the same probability (this is because the number of ancestral nodes adjacent to two extant leaves is precisely the same as the number of ancestral nodes adjacent to no extant leaves, in any such $T$).   In particular, the probability that this process selects an ancestral  node that is
adjacent to a leaf is equal to the proportion of ancestral nodes that are incident with at least  one extant leaf, and this proportion (for any binary tree) is at least  $\frac{1}{2}$.  
Thus, with probability at least $\frac{1}{2}$
the node $v$ selected by this process will have:
\begin{equation}
\label{PAinM}
PA_M(v, T)   \geq  \pi + (1-\pi)  \exp(-b_r \theta),
\end{equation}
where $\theta$ is the length of the pendant edge $e$ that was selected uniformly at random.
On the other hand, if the node selected is not incident with an extant leaf then 
$PA_M(v, T) \geq \pi$, by the assumption that the model is conservative.  So combining this with (\ref{PAinM}) we have, 
conditional on the length  $\theta$ of the randomly selected pendant edge $e$,
$$\overline{PA}_M(T)  \geq \pi + \frac{1}{2} (1-\pi)  \exp(-b_r \theta).$$

Let us now sample a Yule tree with $n$ leaves  that has grown from a single lineage  (this is equivalent, under a uniform improper prior,  to growing the tree until just before it has $n+1$ leaves).  
Then the length of a randomly selected pendant edge in a Yule tree (call it $L$)  has expected value $\frac{1}{2\lambda}$ \citep{moo}.  
Thus, by Jensen's inequality (applied to the convex function $y=\exp(-x)$) the expected value of
$\exp(-b_r L)$ is at least $\exp(-\frac{b_r}{2\lambda})$ and so 
$$\overline{PA}_{M;Y}=\EE[\overline{PA}_M(\T)]   \geq  \pi + \frac{1}{2}(1-\pi)  \exp(-\frac{b_r}{2\lambda}) \geq  \pi + \frac{1}{2}(1-\pi)(1-\frac{b_r}{2\lambda}).$$
Alternatively, we can use the fact that $L$ has an exponential distribution \citep{sta} to obtain a similar bound, 
noting that the expected value of $\exp(-b_rL)$ is exactly $1/(1+ \frac{b_r}{2\lambda}).$
This shows that $\overline{PA}_{M;Y}$ is bounded away from the trivial bound, as claimed.

\newpage

\setcounter{page}{1}
\section{Online Appendix}

\subsection{1.  Proof of Theorem~\ref{firstthm}}

{\em Proof of (1)}  Consider the tree $T_{2n}$ with leaf set $\L_{2n} = \{1,2,...,2n\}$ shown in Fig.~\ref{tricky_tree1}(a).  We can estimate the root state as follows.
Record the frequency of the states that occurs at the odd-labelled leaves ($1,3, \ldots, 2n-1$), and estimate the root state by the most frequently occurring state (in the case of a tie, break it uniformly). 
The reason for ignoring half the leaves is purely for mathematical convenience, since  if we do so, and if the root is in state $i$, then for each $j$  the number $N_j$ of odd-labelled leaves that are in state $j$
has a binomial distribution with $n$ trials and probability of success of $\pi_{ij}(t)$ where $\pi_{ij}(t)$ is the transition probability from state $i$ to state $j$ in time $t$.

It follows from the Central Limit Theorem (or, indeed, just the weak law of large numbers) that the values $N_j/n$ converge in probability to the corresponding $\pi_{ij}(t)$ values, and so the probability that the majority state agrees with the root state converges to 1 as $n$ grows (provided $t$ is held fixed) since for a conservative model $\pi_{ii}(t) > \pi_{ij}(t)$ for all $j \neq i$.

Let us now consider an internal non-root node $v_i$.  In this case,
for any  subset $U=\{u_1,\ldots, u_k\}$ of nodes of $T_{2n}$, let $X(u_1, \ldots, u_k)$ denotes the random variable that describes the assignment of states to the nodes in $U$.  Then we have the following Markov chain:
$$X(v_i) \longrightarrow X(2i-1, 2i, \rho) \longrightarrow X(\L_{2n}),$$
since $X(v_i)$ and $X(\L_{2n})$ are conditionally independent once we specify the states at the root node and the leaves $2i-1$ and $2i$.  
We now invoke the Data-Processing Inequality \citep{cov} to  obtain:
\begin{equation}
\label{loss}
I(X(v_i); X(\L_{2n})) \leq I(X(v_i); X(2i-1, 2i, \rho)).
\end{equation}
Now, $I(X(v_i); X(2i-1, 2i, \rho))$ is independent of $n$, and decays to zero exponentially fast with increasing $t$ by Proposition 6 of \citet{sob}.   Thus,  from (\ref{loss})  we have:
$$I(X(v_i); X(\L_{2n})) \leq c_1\exp(-c_2t),$$ for constants $c_1, c_2>0$.

It follows by Lemma~\ref{rab} that, for any $\delta>0$, we can select $t$ and $n$ appropriately so that  the predictive accuracy for any method for inferring the non-root internal nodes is less than $\delta$, and simultaneously (by the first part of the proof) the predictive accuracy for the root is at least $1-\delta$.  This completes the proof of part  (1).

{\em Proof of (2)}  Consider the tree in shown in Fig.~\ref{tricky_tree1}(b).  Suppose the height of the two pendant trees is at most $\epsilon$ and the length of the two branches incident with the root is at least $L$.  Then  the states at the nodes in the two subtrees
can be predicted with an accuracy that converges to 1 as $\epsilon$ converges towards 0 (by e.g. majority rule applied to the leaves below $v$ since as $\epsilon$ converges to zero the probability that all the leaves are in the same state as $v$ converges to 1).   On the other hand,  by Lemma~\ref{rab}, and Proposition 6 of \cite{sob} the accuracy of  inferring the root state
converges to the trivial bound as $L \rightarrow \infty$.

\subsection{2.  Proof of Theorem~\ref{nextthm}}  

{\em Part (1)}:  Consider the tree in  Fig.~\ref{tricky_tree2}(a).  We will suppose that the left-hand tree $T_0$ is (close to) a star tree, and the right-hand tree consists of two fully balanced binary trees, each having edges of equal length.   We first consider the predictive accuracy of MR.

For any $\delta>0$ we can select the length $L$ of the branch leading to $T_0$  sufficiently long so that the probability that the root $v$ of $T_0$  has the same state as the root node $\rho$ of $T$ is no more than $\frac{1}{2} + \delta$.  Moreover, for $n$ sufficiently large, and $\epsilon$ (the height of the left-hand tree) sufficiently short, the number of leaves in $T_0$ in the same state as node $v$  is at least $\frac{3}{4}n^2$ with probability as close to 1 as we wish. 

Thus, since $T_1$ and $T_2$ have between them just $2n$ leaves, 
when $v$ has the same state as $\rho$ (and $n>8$) it follows that  MR will correctly predict the root state of $T$, but when $v$ has a different state to the root state it will not.  The probability of the former event is $\frac{1}{2} + \delta$, and so $PA_{MR}(\rho, T) \leq \frac{1}{2} + \delta.$

For the claim concerning MP, we first establish a general lemma concerning the predictive accuracy of this method when applied to a chain of subtrees.

Consider the tree shown in Fig.~\ref{tricky_tree2}(a) obtained by making rooted binary subtrees $T_0, T_1, T_2$ adjacent to two adjacent nodes
$v_1, v_2$ by attaching $T_0$ and $T_1$ by new edges to $v_1$, and attaching $T_2$ by a new edge to $v_2$. This produces
a rooted binary tree $T(T_0, T_1, T_2)$ with root node $\rho = v_2$. 

\begin{lemma}
\label{lemsim}
In a two-state model, suppose the predictive accuracy of MP for estimating the root of $T_i$ from the leaves of $T_i$ (for $i=1$ and $i=2$) is at least $1-\nu$.
 Then, for any $x>0$, if the  three branches connecting $T_1$ and $T_2$ to the root are sufficiently short,
  then the predictive accuracy of MP for estimating the root state of  $T(T_0, T_1, T_2)$ is at least  $1-4 \nu - x$.
\end{lemma}
\begin{proof}
Without loss of generality suppose that the root state is 0.  A sufficient condition for MP to correctly infer the root state, is that the following three conditions hold: 
\begin{itemize}
\item[(i)] there is no change of state on the three edges of $T(T_0, T_1, T_2)$ that connect $T_1$ and $T_2$ to the root of $T$; and
\item[(ii)]  the first pass of the Fitch algorithm  returns the singleton state $0$ for the root of $T_1$ when applied to $T_1$; and
\item[(iii)] the first pass of the Fitch algorithm  returns the singleton state $0$ for the root of $T_2$ when applied to $T_2$. 
\end{itemize}

When  conditions (i)--(iii) hold it follows that MP will infer the root state as the singleton state $0$ (which matches the true root state) regardless of what the Fitch algorithm returns for $T_0$
and of whether or not  there has been any state change on its incident edge.

 Now, event (i) occurs with probability  at least $1-x$  by making these edges sufficiently short,
 and events (ii) and (iii) are conditionally independent given event (i) and each of them occurs
 with probability at least $(1-2\nu)$, conditional on event (i).  This is because, when we have two states, simple algebra shows that the predictive accuracy of MP in estimating the root state of a tree (in this case $T_i$)  is at most
 $\frac{1}{2}(1+S)$ where $S$ is the probability that the first pass of the Fitch algorithm returns the singleton root state for that subtree). 

Thus,  the probability of the sufficient condition (i)+(ii) +(iii)  is  at least: 
 $(1-2\nu)^2(1-x)  > 1-4\nu-x$.  Thus $PA_{MP}(T, \rho) \geq 1- 4\nu -x.$
\end{proof}

Now, by  Proposition 2.1 of \citet{gas},  the predictive accuracy of MP for the root state of each of $T_1$ and $T_2$ can be made as close to 1 as we wish by ensuring that the equal branch lengths in the right-hand tree are sufficiently small. This will require selecting the number of leaves $n$ sufficiently large.  Lemma~\ref{lemsim} then ensures that MP will infer the root state of $T$ with an accuracy that can be made as close to 1 as we wish.

{\em Part (2)}: Consider a caterpillar tree with $n$ leaves, in which the interior edges of the tree lie all within height $\epsilon$ from the root as indicated in Fig.~\ref{tricky_tree2}(b).  As $\epsilon$ tends to zero, the 2-state symmetric process converges to a process  on the caterpillar tree with $n$ leaves in which the leaves are assigned the root state ($0$) independently with probability $q= 1-p$, and the state $1$ with probability $p$.  Let $s_n$ (respectively $d_n$ and $e_n$) denote the probability that, for this assignment of states to the leaves, the most parsimonious reconstruction assigns the root state
$0$ (respectively, state $1$ and the indeterminate state $\{0,1\}$).  Then the values $(s_n, d_n, e_n)$ satisfy the linear recursion:
$s_1= q, d_1=p, e_1=0$, and for $n \geq 1$,
$$s_{n+1} = qs_n + qe_n;$$
$$d_{n+1} = pd_n + pe_n;$$
$$e_{n+1} = ps_n + qd_n.$$
Let $(s,d,e) =  \lim_{n \rightarrow \infty} (s_n, d_n, e_n)$
(this limit exists, and is even independent of the initial values for $s(1), d(1), e(1)$ because the linear recursion described corresponds to the  transition matrix of a 3-state Markov chain which is irreducible and aperiodic, and so it converges to a unique equilibrium distribution, regardless of the distribution of the initial state). 
Then $$(s, d, e)  = \frac{1}{1-pq}(q^2, p^2, pq).$$
Consequently the accuracy of parsimony on this tree is: 
$$s+ \frac{1}{2}e = \frac{q^2+pq/2}{1-pq}.$$
Notice that this function is a continuous function of $p$, it takes the value $\frac{1}{2}$ at $q=0.5$ and it increases monotonically to the value $1$ at $q=1$.
Thus, we can select $q>0.5$ so that $s+ \frac{1}{2}e < \frac{1}{2}+ \delta$.  It follows that (for $\epsilon$ sufficiently small, and $n$ sufficient large) the predictive accuracy 
for MP in root state estimation will be less than the trivial bound ($\frac{1}{2}$) plus $\delta$.
On the other hand, for MR, the number of leaves in the same state ($0$) as the root is a sum of $n$ identically-distributed Bernoulli random variables, which become 
independent as $\epsilon \rightarrow 0$. Thus, by the central limit theorem we can select $n$ sufficiently large (and $\epsilon$ sufficiently small) so that the predictive accuracy 
of MR is at least $1-\delta$.

\subsection{3.  Proof of Proposition~\ref{parsimony_yule}}

The proof relies on the following result.

\begin{lemma}
\label{next3thm}
\mbox{ } 

Suppose that a rooted binary tree $\T$ on $n=2^h$ leaves is completely balanced (i.e. each leaf is the same number of edges from the root).  If each leaf is independently assigned the root state $0$ with  probability $1-p> \frac{1}{2}$ then the probability under a two-state symmetric model that the most parsimonious state at the root is $0$ converges to $1$ as $n$ tends to infinity.  
\end{lemma}

Before presenting this proof of this lemma, we show how it implies Proposition~\ref{parsimony_yule}.  
 For the tree shown in Fig.~\ref{tricky_tree2}(b) if the top portion of the tree is completely balanced  then, as $\epsilon$ converges to zero, and $n$ grows (the height of the tree is fixed), the predictive accuracy of MP in estimating the root state in a symmetric two-state model converges to 1 (note  that as $\epsilon$ converges to zero, all states at the nodes of the top portion of the tree agree with the root state, which allows us to apply Lemma~\ref{next3thm}).

\subsection{Proof of Lemma~\ref{next3thm}}

Consider a completely balanced tree of depth $h$ (and so with $n=2^h$ leaves), in which each leaf is assigned state $0$ (the root state) independently with probability $1-p$ and 
state 1 with probability $p$.  Let $s_h$ (respectively $d_h, e_h$) be the probability that for the resulting assignment of states to the leaves of the tree
the most parsimonious reconstruction of the root state is $0$ (respectively,  $1, \{0,1\})$.
We have $s_1= (1-p)^2, d_1 = p^2, e_1 = 2p(1-p)$, and for $h \geq 1$:
$$s_{h+1} = s_h^2 + 2s_he_h;$$
$$d_{h+1} = d_h^2 + 2d_he_h;$$
$$e_{h+1} = e_h^2 + 2s_hd_h.$$

Subtracting the second equation above from the first  (and using the identity $s_h+d_h = 1-e_h$) we obtain the following equation for $G_h = s_h -d_h$:
$$ G_{h+1} = G_h(1-e_h) + 2G_h e_h,$$ which we can rewrite as:
\begin{equation}
\label{Geq}
G_{h+1} = G_h(1+e_h).
\end{equation}
Notice that if $p<1/2$ then $G_1 >0$ in which case Eqn. (\ref{Geq}) implies that  $G_h$ is monotone increasing. Thus, since $G_h$ is also bounded above by $1$,  $G_h$ has a limit. Moreover, the convergence of $G_h$ (together with  Eqn.~(\ref{Geq})) now implies  that $e_h \rightarrow 0$ as $h \rightarrow \infty$.
Since $e_h = 1-(s_h+d_h)$ this implies that $s_h+d_h$ converges to $1$, and that $s_h-d_h (=G_h)$ converges to some constant, and hence $(s_h, d_h, e_h)$ has a limit $(s, d, 0)$.
From the two recursion equations for $s_h, d_h$ above we see that $s$ and $d$ must
satisfy the system $s=s^2; d=d^2,$ and since  $s-d \geq G_1>0$ the only solution possible is 
$s=1, d=0$.  Thus,  the predictive accuracy of MP converges to $1$ as $h \rightarrow \infty$, which completes the proof of Lemma~\ref{next3thm}.

We finish by stating the analogous recursion when the trees on $n$ leaves are generated according to the Yule-Harding distribution. Again assume each leaf is assigned state $0$ (the root state) independently with probability $1-p$ and 
state 1 with probability $p$, and let $S(n)$ (respectively, $D(n)$, $E(n)$)  be the probability that for the resulting assignment of states to the leaves of the tree the most parsimonious reconstruction of the root
state is $0$ (respectively,  $1, \{0,1\})$.   Since the Yule-Harding distribution provides a uniform distribution on the size distribution of the two subtrees incident with the root (and each of these two
subtrees follows the Yule-Harding distribution) the recursion for 
$(S(n), D(n), E(n))$ is as follows:
$S(1)=1-p, D(1)= p, E(1) = 0,$ and for $n >1$:
$$S(n) = \frac{1}{n-1}\sum_{i=1}^{n-1}[ S(i)S(n-i) + 2S(i)E(n-i)],$$
$$D(n) =  \frac{1}{n-1}\sum_{i=1}^{n-1} [D(i)D(n-i) + 2D(i)E(n-i)],$$
$$E(n) =  \frac{1}{n-1}\sum_{i=1}^{n-1} [E(i)E(n-i) + 2S(i)D(n-i)].$$

\hfill$\Box$

\newpage

\subsection{4. Additional simulation results with Yule trees and the two-state symmetric model}

\subsection{4.1 Method accuracy}

\noindent The accuracy of all tested methods is displayed in the following tables, for 10, 100 and 1,000 tips, respectively. We measure the accuracy in reconstructing (1) the root state, (2) the state of any randomly selected ancestral node (including the tree root), and (3) the changes along the interior branches. See text for details on the simulation procedure, the tested methods and the accuracy measures. 

\vspace{2cm}
\newpage
\noindent \textbf{10 TIPS}
\vspace{1cm}

\begin{tabular}{|l|c|c|c|c|c|c|c|c|c|} \hline 
\textbf{Rate ratio} & \textbf{1} & \textbf{2} & \textbf{3} & \textbf{4} & \textbf{5} & \textbf{6} & \textbf{8} & \textbf{12} & \textbf{20} \\ \hline 
Root-Majority & 0.542 & 0.655 & 0.728 & 0.794 & 0.829 & 0.852 & 0.889 & 0.925 & 0.956 \\ \hline 
Root-Parsimony & 0.535 & 0.641 & 0.710 & 0.781 & 0.811 & 0.840 & 0.878 & 0.917 & 0.952 \\ \hline 
Root-Likelihood & 0.545 & 0.665 & 0.745 & 0.815 & 0.849 & 0.873 & 0.908 & 0.941 & 0.968 \\ \specialrule{0.15em}{0em}{0em}
Node-Majority & 0.700 & 0.807 & 0.860 & 0.891 & 0.911 & 0.926 & 0.944 & 0.962 & 0.978 \\ \hline 
Node-Parsdown & 0.692 & 0.796 & 0.849 & 0.883 & 0.902 & 0.918 & 0.938 & 0.957 & 0.975 \\ \hline 
Node-Parsacct & 0.706 & 0.822 & 0.880 & 0.915 & 0.933 & 0.947 & 0.965 & 0.980 & 0.991 \\ \hline 
Node-Parsdelt & 0.699 & 0.815 & 0.873 & 0.909 & 0.930 & 0.945 & 0.963 & 0.979 & 0.990 \\ \hline 
Node-Parsinde & 0.699 & 0.816 & 0.873 & 0.910 & 0.930 & 0.944 & 0.963 & 0.979 & 0.990 \\ \hline 
Node-Likedown & 0.703 & 0.813 & 0.866 & 0.898 & 0.918 & 0.932 & 0.950 & 0.967 & 0.981 \\ \hline 
Node-Likmargi & 0.714 & 0.832 & 0.888 & 0.920 & 0.939 & 0.952 & 0.967 & 0.980 & 0.990 \\ \hline 
Node-Likebest & 0.710 & 0.831 & 0.888 & 0.922 & 0.942 & 0.955 & 0.971 & 0.984 & 0.993 \\ \specialrule{0.15em}{0em}{0em}
Branch-Majority & 0.515 & 0.665 &  0.753 &  0.804 & 0.841 &  0.861 & 0.897 & 0.932 & 0.957 \\ \hline 
Branch-Parsdown & 0.490 & 0.633 &  0.722 &  0.775 & 0.817 &  0.837 &  0.879 & 0.918 & 0.948 \\ \hline 
Branc -Parsacct & 0.544 & 0.703 &  0.799 &  0.852 & 0.890 &  0.906 &  0.941  & 0.966 & 0.982 \\ \hline 
Branch-Parsdelt & 0.539 & 0.695 &  0.793 &  0.848 & 0.888 &  0.906 &  0.940 & 0.965 & 0.983 \\ \hline 
Branch-Parsinde & 0.534 & 0.696 &  0.792 &  0.848 & 0.887 &  0.905 &  0.940 & 0.966 & 0.982 \\ \hline 
Branch-Likedown & 0.521 & 0.673 &  0.762 &  0.815 & 0.851 &  0.869 &  0.907 & 0.940 & 0.962 \\ \hline 
Branch-Likmargi & 0.552 & 0.713 &  0.805 &  0.857 & 0.890 &  0.905 &  0.939 & 0.963 & 0.978 \\ \hline 
Branch-Likebest & 0.563 & 0.726 &  0.820 &  0.872 & 0.907 &  0.922 &  0.953 & 0.974 & 0.987 \\ \hline 
\end{tabular}

\vspace{2cm}
\newpage
\noindent \textbf{100 TIPS}
\vspace{1cm}

\begin{tabular}{|l|c|c|c|c|c|c|c|c|c|} \hline 
\textbf{Rate ratio} & \textbf{1} & \textbf{2} & \textbf{3} & \textbf{4} & \textbf{5} & \textbf{6} & \textbf{8} & \textbf{12} & \textbf{20} \\ \hline 
Root-Majority & 0.503 & 0.541 & 0.625 & 0.702 & 0.748 & 0.802 & 0.851 & 0.905 & 0.944 \\ \hline 
Root-Parsimony & 0.503 & 0.533 & 0.601 & 0.667 & 0.722 & 0.776 & 0.831 & 0.894 & 0.943 \\ \hline 
Root-Likelihood & 0.504 & 0.547 & 0.642 & 0.725 & 0.781 & 0.832 & 0.880 & 0.933 & 0.966 \\ \specialrule{0.15em}{0em}{0em}
Node-Majority & 0.690  & 0.793  & 0.848  & 0.883  & 0.906  & 0.921  & 0.942  & 0.962  & 0.977  \\ \hline 
Node-Parsdown & 0.681  & 0.780  & 0.836  & 0.871  & 0.896  & 0.912  & 0.935  & 0.957  & 0.974  \\ \hline 
Node-Parsacct & 0.696  & 0.809  & 0.872  & 0.911  & 0.936  & 0.951  & 0.970  & 0.985  & 0.994  \\ \hline 
Node-Parsdelt & 0.685  & 0.797  & 0.865  & 0.907  & 0.934  & 0.952  & 0.971  & 0.986  & 0.995  \\ \hline 
Node-Parsinde & 0.688  & 0.800  & 0.866  & 0.907  & 0.933  & 0.950  & 0.970  & 0.985  & 0.994  \\ \hline 
Node-Likedown & 0.692  & 0.798  & 0.857  & 0.892  & 0.915  & 0.930  & 0.949  & 0.968  & 0.981  \\ \hline 
Node-Likmargi & 0.704  & 0.824  & 0.889  & 0.926  & 0.949  & 0.963  & 0.978  & 0.989  & 0.996  \\ \hline 
Node-Likebest & 0.697 & 0.816 & 0.884 & 0.923 & 0.947 & 0.962  & 0.978 & 0.990 & 0.996 \\ \specialrule{0.15em}{0em}{0em}
Branch-Majority &  0.488 & 0.639 & 0.730 & 0.786 & 0.826 & 0.854 & 0.890 & 0.924 & 0.956 \\ \hline 
Branch-Parsdown &  0.460 & 0.603 & 0.695 & 0.753 & 0.796 & 0.828 & 0.868 & 0.909 & 0.947 \\ \hline 
Branc -Parsacct &  0.511 & 0.678 & 0.781 & 0.842 & 0.883 & 0.912 & 0.944 & 0.970 & 0.989 \\ \hline 
Branch-Parsdelt &  0.504 & 0.667 & 0.775 & 0.840 & 0.884 & 0.914 & 0.947 & 0.973 & 0.990 \\ \hline 
Branch-Parsinde &  0.492 & 0.658 & 0.766 & 0.831 & 0.877 & 0.908 & 0.942 & 0.970 & 0.989 \\ \hline 
Branch-Likedown &  0.491 & 0.647 & 0.742 & 0.799 & 0.840 & 0.868 & 0.903 & 0.936 & 0.964 \\ \hline 
Branch-Likmargi &  0.527 & 0.703 & 0.809 & 0.868 & 0.906 & 0.931 & 0.958 & 0.978 & 0.992 \\ \hline 
Branch-Likebest &  0.530 & 0.701 & 0.807 & 0.867 & 0.907 & 0.932 & 0.959 & 0.980 & 0.993 \\ \hline 
\end{tabular}

\vspace{2cm}
\newpage
\noindent \textbf{1,000 TIPS}
\vspace{1cm}

\begin{tabular}{|l|c|c|c|c|c|c|c|c|c|} \hline 
\textbf{Rate ratio} & \textbf{1} & \textbf{2} & \textbf{3} & \textbf{4} & \textbf{5} & \textbf{6} & \textbf{8} & \textbf{12} & \textbf{20} \\ \hline 
Root-Majority & 0.500 & 0.511 & 0.580 & 0.665 & 0.732 & 0.772 & 0.836 & 0.898 & 0.940 \\ \hline 
Root-Parsimony & 0.496 & 0.507 & 0.544 & 0.609 & 0.678 & 0.734 & 0.806 & 0.882 & 0.935 \\ \hline 
Root-Likelihood & 0.503 & 0.513 & 0.590 & 0.681 & 0.759 & 0.807 & 0.870 & 0.927 & 0.963 \\ \specialrule{0.15em}{0em}{0em}
Node-Majority & 0.690  & 0.790  & 0.848  & 0.883  & 0.905  & 0.921  & 0.941  & 0.961  & 0.977  \\ \hline 
Node-Parsdown & 0.681  & 0.777  & 0.834  & 0.870  & 0.895  & 0.911  & 0.934  & 0.956  & 0.974  \\ \hline 
Node-Parsacct & 0.696  & 0.806  & 0.871  & 0.911  & 0.935  & 0.951  & 0.970  & 0.985  & 0.994  \\ \hline 
Node-Parsdelt & 0.684  & 0.794  & 0.864  & 0.908  & 0.935  & 0.953  & 0.972  & 0.987  & 0.995  \\ \hline 
Node-Parsinde & 0.688  & 0.797  & 0.865  & 0.907  & 0.934  & 0.951  & 0.970  & 0.986  & 0.995  \\ \hline 
Node-Likedown & 0.692  & 0.796  & 0.856  & 0.891  & 0.914  & 0.930  & 0.949  & 0.968  & 0.981  \\ \hline 
Node-Likmargi & 0.704  & 0.821  & 0.889  & 0.928  & 0.951  & 0.965  & 0.979  & 0.991  & 0.997  \\ \hline 
Node-Likebest & 0.697 & 0.812 & 0.883 & 0.925 & 0.949 & 0.963 & 0.979 & 0.990 & 0.997 \\ \specialrule{0.15em}{0em}{0em}
Branch-Majority &  0.485 &  0.634 & 0.726 &  0.784 & 0.823 & 0.851 & 0.888 & 0.925 & 0.955 \\ \hline 
Branch-Parsdown &  0.458 &  0.598 & 0.690 &  0.750 & 0.793 & 0.824 & 0.866 & 0.909 & 0.946 \\ \hline 
Branc -Parsacct &  0.509 &  0.673 & 0.776 &  0.840 & 0.882 & 0.910 & 0.943 & 0.972 & 0.989 \\ \hline 
Branch-Parsdelt &  0.502 &  0.663 & 0.770 &  0.840 & 0.885 & 0.914 & 0.949 & 0.976 & 0.991 \\ \hline 
Branch-Parsinde &  0.488 &  0.649 & 0.757 &  0.827 & 0.874 & 0.905 & 0.941 & 0.972 & 0.989 \\ \hline 
Branch-Likedown &  0.488 &  0.642 & 0.738 &  0.798 & 0.838 & 0.866 & 0.901 & 0.936 & 0.963 \\ \hline 
Branch-Likmargi &  0.525 &  0.699 & 0.808 &  0.871 & 0.910 & 0.934 & 0.961 & 0.982 & 0.993 \\ \hline 
Branch-Likebest &  0.528 &  0.695 & 0.803 &  0.868 & 0.909 & 0.933 & 0.961 & 0.982 & 0.993 \\ \hline 
\end{tabular}

\newpage
\subsection{4.2 Results with biased sampling}

\noindent In these experiments we use the same (Yule) trees and (two-state, symmetric) data, but one of the character state (say 0) has probability 0.5 to be sampled, while the other (i.e. 1) keeps a probability 1.0 to be sampled. All other parameters remain identical. To simulate this condition we replace half of the 0s at the tree tips by an unknown character, and then the tested methods are launched in the standard way. Results with 100 tips are provided in Table below, with a speciation/substitution rate ratio of 6.

\vspace{2cm}

\begin{tabular}{|l|c|c|} \hline 
\textbf{} & \textbf{Root} & \textbf{Node} \\ \hline 
Majority & 0.700 & 0.876 \\ \hline 
Parsdown & 0.737 & 0.874 \\ \hline 
Parsacct & 0.738 & 0.915 \\ \hline 
Parsdelt & 0.736 & 0.899 \\ \hline 
Parsinde & 0.739 & 0.904 \\ \hline 
Likedown & 0.784 & 0.860 \\ \hline 
Likmargi & 0.784 & 0.924 \\ \hline 
Likebest & 0.786 & 0.925 \\ \hline 
\end{tabular}

\newpage
\subsection{5. Additional simulation results with non-molecular clock trees and the HKY+Gamma substitution model}

\noindent The accuracy of all tested methods is displayed in the following table. We measure the accuracy in reconstructing: (1) the root state, (2) the state of any randomly selected ancestral node (including the tree root), and (3) the changes along the interior branches. We also provide (4) the square of node accuracy (in parentheses), for comparison with branch accuracy assuming predictions are independent at both branch extremities. See text for details on the simulation procedure, the tested methods and the accuracy measures.

\resizebox{!}{11cm}{
\begin{tabular}{|l|c|c|c|c|c|} \hline 
\textbf{} & \textbf{Tips} & \textbf{Root} & \textbf{Node} & \textbf{Branch} & \textbf{(square)} \\ \hline 
Majority & {\bf 25}  & 0.833 & 0.919 & 0.858 & (0.845) \\ \hline 
Parsdown & {\bf 25} & 0.822 & 0.911 & 0.836 & (0.831) \\ \hline 
Parsacct &  {\bf 25} & 0.821 & 0.945 & 0.908 & (0.893) \\ \hline 
Parsdelt & {\bf 25} & 0.820 & 0.942 & 0.906 & (0.888) \\ \hline 
Parsinde & {\bf 25} & 0.821 & 0.942 & 0.899 & (0.887) \\ \hline 
Likedown & {\bf 25}& 0.874 & 0.939 & 0.887 & (0.881) \\ \hline  
Likemargi & {\bf 25} & 0.874 & 0.958 & 0.925 & (0.918) \\ \hline  
Likemargi-JC & {\bf 25} & 0.872 & 0.956 & 0.923 & (0.915) \\ \specialrule{0.2em}{0em}{0em}
Majority & {\bf 50} & 0.845 & 0.930 & 0.874 & (0.864) \\ \hline 
Parsdown & {\bf 50} & 0.834 & 0.923 & 0.855 & (0.853) \\ \hline 
Parsacct & {\bf 50} & 0.834 & 0.957 & 0.925 & (0.915) \\ \hline 
Parsdelt & {\bf 50} & 0.834 & 0.954 & 0.923 & (0.910) \\ \hline 
Parsinde & {\bf 50} & 0.835 & 0.954 & 0.919 & (0.910) \\ \hline 
Likedown & {\bf 50} & 0.885 & 0.948 & 0.903 & (0.898) \\ \hline 
Likemargi & {\bf 50} & 0.885 & 0.968 & 0.943 & (0.938) \\ \hline 
Likemargi-JC & {\bf 50}   & 0.881 & 0.966 & 0.940 & (0.934) \\ \specialrule{0.2em}{0em}{0em}
Majority & {\bf 100} & 0.867 & 0.943 & 0.897 & (0.889) \\ \hline 
Parsdown & {\bf 100} & 0.852 & 0.937 & 0.879 & (0.878) \\ \hline 
Parsacct & {\bf 100} & 0.854 & 0.968 & 0.944 & (0.937) \\ \hline 
Parsdelt & {\bf 100} & 0.853 & 0.966 & 0.943 & (0.934) \\ \hline 
Parsinde & {\bf 100} & 0.853 & 0.966 & 0.939 & (0.934) \\ \hline 
Likedown & {\bf 100} & 0.907 & 0.959 & 0.922 & (0.919) \\ \hline 
Likemargi & {\bf 100} & 0.907 & 0.978 & 0.960 & (0.957) \\ \hline 
Likemargi-JC & {\bf 100} & 0.905 & 0.976 & 0.957 & (0.954) \\ \specialrule{0.2em}{0em}{0em}
Majority & {\bf 200} & 0.872 & 0.949 & 0.907 & (0.901) \\ \hline 
Parsdown & {\bf 200} & 0.864 & 0.944 & 0.890 & (0.890) \\ \hline 
Parsacct & {\bf 200} & 0.866 & 0.973 & 0.953 & (0.947) \\ \hline 
Parsdelt & {\bf 200} & 0.865 & 0.972 & 0.952 & (0.945) \\ \hline 
Parsinde & {\bf 200} & 0.866 & 0.972 & 0.949 & (0.945) \\ \hline 
Likedown & {\bf 200} & 0.917 & 0.964 & 0.930 & (0.929) \\ \hline 
Likemargi & {\bf 200} & 0.917 & 0.983 & 0.968 & (0.965) \\ \hline  
Likemargi-JC & {\bf 200} & 0.913 & 0.982 & 0.966 & (0.965) \\ \specialrule{0.2em}{0em}{0em}
Majority & {\bf 400} & 0.886 & 0.955 & 0.918 & (0.913) \\ \hline  
Parsdown & {\bf 400} & 0.883 & 0.951 & 0.903 & (0.904) \\ \hline 
Parsacct & {\bf 400} & 0.883 & 0.978 & 0.961 & (0.957) \\ \hline  
Parsdelt & {\bf 400} & 0.883 & 0.977 & 0.960 & (0.955) \\ \hline 
Parsinde &  {\bf 400} & 0.884 & 0.977 & 0.958 & (0.955) \\ \hline 
Likedown & {\bf 400}  & 0.929 & 0.969 & 0.939 & (0.939) \\ \hline 
Likemargi & {\bf 400} & 0.929 & 0.986 & 0.975 & (0.973) \\ \hline  
Likemargi-JC & {\bf 400} & 0.925 & 0.984 & 0.973 & (0.969) \\ \hline  
\end{tabular}
}

\end{document}